\documentclass{archive}
\usepackage{graphicx,amssymb,amsmath, algorithmic}
\newif\ifFull
\Fullfalse

\ifFull\else
\renewcommand{\subsection}[1]{\paragraph{\textbf{#1}.}}
\fi

\newcommand{\R}{\mathbb{R}}  
\newif\ifFull
\Fullfalse

\title{Geometric Fingerprint Recognition via
Oriented Point-Set Pattern Matching}

\author{David Eppstein \thanks{University of California, Irvine}
\and
Michael T.~Goodrich\footnotemark[1]
\and
Jordan Jorgensen\footnotemark[1]
\and
Manuel R.~Torres \thanks{University of Illinois}}
 
\index{Jorgensen, Jordan}
\index{Torres, Manuel R.}
\index{Goodrich, Michael T.}
\index{Eppstein, David}

\begin{document}

\maketitle

\begin{abstract}
Motivated by the problem of fingerprint matching,
we present geometric approximation algorithms 
for matching a \emph{pattern} point
set against a \emph{background} point set, where the points 
have angular orientations in addition to their positions. 
\ifFull
We define such matching problems in terms 
of minimizing a directed Hausdorff distance between a pair of such 
oriented point sets, 
based on an underlying metric that combines 
positional distance and angular distance.
We present a family of fast approximation algorithms for such
oriented point-set pattern matching problems that are based on simple
pin-and-query and grid-refinement strategies.
Our algorithms achieve an approximation ratio of $1+\epsilon$, for any fixed
constant $\epsilon>0$.
\fi
\end{abstract}


\section{Introduction}

Fingerprint recognition
typically involves a three-step process: (1) digitizing fingerprint images, 
(2) identifying \emph{minutiae},
which are
points where ridges begin, end, split, or join,
and (3) matching corresponding minutiae points between the two 
images.
An important consideration is that the minutiae are not pure geometric points: besides having geometric positions, defined by $(x,y)$ coordinates
in the respective images, each minutiae point also has 
an \emph{orientation} (the direction of the associated ridges), 
and these orientations
should be taken into consideration in the
comparison, e.g., see~\cite{maltoni2009handbook,jain1997identity,%
ratha2007automatic,Xu09,JEA20051672,jiang906252,%
tico2003fingerprint,qi2005fingerprint,kulkarni2006orientation}
and Figure~\ref{fig:fingerprint}.

\begin{figure}[htb]
\centering
\ifFull
\includegraphics[width=3.2in]{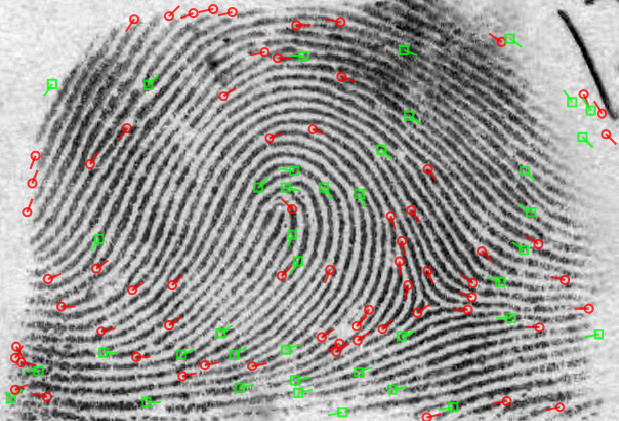}
\else
\includegraphics[width=2.5in]{figs/fingerprint.png}
\fi
\caption{Screenshot of the display of fingerprint minutiae in NIST's Fingerprint Minutiae Viewer (FpMV). 
\ifFull
Public domain U.S.~government image.
\fi}
\label{fig:fingerprint}
\end{figure}

In this paper, we consider computational geometry problems inspired
by this fingerprint matching problem.
The problems we consider are all instances of
point-set pattern matching problems, where we are given a 
``pattern'' set, $P$, of $m$ points in $\R^2$ and a 
``background'' set, $B$, of $n$ points in $\R^2$,
and we are asked to
find a transformation of $P$ that best aligns the points of $P$ with 
a subset of the points in $B$,
e.g., 
see~\cite{cardoze1998pattern,chew1997geometric,cho2008improved,%
gavrilov1999geometric,goodrich1999approximate}.
\ifFull
The point-set pattern matching problem 
has many applications,
including for image matching~\cite{he2003image,jain1997identity}, computer 
vision~\cite{mount1999efficient}, image compression~\cite{alzina20022d},
computational chemistry~\cite{finn1997rapid}, and 
touch input recognition~\cite{kristensson2005relaxing}, among others.
\fi

A natural choice of a distance measure to use in this case, 
between a transformed copy, $P'$, of 
the pattern, $P$, against the background, $B$,
is the \emph{directed Hausdorff distance},
defined as 
$
h(P', B) = \max_{p \in P'}\min_{q \in B}\rho(p, q)
$,
where 
$\rho$ is an underlying distance metric for points,
such as the Euclidean metric. 
In other words, the problem is to find a transformation of 
$P$ that minimizes the farthest any point in $P$ is from its nearest neighbor 
in $B$.
Rather than only considering the positions of the points in $P$ and
$B$, however, in this paper we consider instances
in which each point in $P$ and $B$ also has
an associated \emph{orientation} defined by an angle, as in
the fingerprint matching application.

It is important in such \emph{oriented point-set pattern matching} problems
to use an underlying distance that combines information about both the locations
and the orientations of the points, and to use this distance in finding a good
transformation.
Our goal is to design efficient
algorithms that can find a transformation that is a good match between
$P$ and $B$ taking both positions and orientations into
consideration.

\ifFull
There are several challenges to designing such 
algorithms, however. 
Adding orientation to point-set pattern matching further
complicates a computational geometry problem
whose exact solution
is already difficult, both in terms of geometric complexity (involving
computing intersections of non-linear curves and surfaces) and algorithmic
complexity (involving high running times in terms of $n$ and $m$).
Thus, we are interested in this paper in \emph{simple} approximation algorithms that
can achieve guaranteed high-quality performance in comparison to the theoretically optimal
(but harder to compute) solutions.
\fi

\subsection{Previous Work}
In the domain of fingerprint matching, past work tends to focus on matching
fingerprints heuristically or as pixelated images, 
taking into consideration both the positions
and orientation of the minutiae or
other features,
e.g., see~\cite{maltoni2009handbook,jain1997identity,%
ratha2007automatic,Xu09,JEA20051672,jiang906252,%
tico2003fingerprint,qi2005fingerprint,kulkarni2006orientation}.
We are not aware of past work on studying fingerprint matching
as a computational geometry problem, however.
\ifFull
That is, we are not familiar with previous work
on fingerprint matching as a geometric 
optimization problem of oriented point-set pattern matching with the goal of
minimizing a directed Hausdorff distance.
\fi

Geometric pattern matching for point sets without orientations, on the 
other hand, has been well studied from a computational geometry
viewpoint, 
e.g., see~\cite{alt1999discrete,chew1997geometric,gavrilov1999geometric,veltkamp2001shape}.
For such unoriented point sets,
existing algorithms can find an optimal solution
minimizing Hausdorff distance, 
but they generally have high polynomial running times. 
\ifFull
For example,
one such solution for planar point sets 
is due
to Chew {\it et al.}~\cite{chew1997geometric} and has a running time of 
$O(n^2m^3\log^2{nm})$. 
Their algorithm, which is based on an exact construction of a high-dimensional
configuration space, can be extended to oriented point-set pattern matching, but
it would lead to complex algorithms with even higher running times.
Thus, point-set
pattern matching is a natural domain where
we desire simple approximation algorithms.

\fi
Several existing algorithms give approximate solutions 
to geometric pattern matching
problems~\cite{cardoze1998pattern,cho2008improved,goodrich1999approximate,indyk1999geometric},
but we are not aware of previous approximation algorithms
for oriented point-set pattern matching.
\ifFull
In the context of unoriented points,
an approximation algorithm is usually characterized by its approximation ratio
to the optimal Hausdorff distance,
that is,
the largest factor by which the Hausdorff distance 
achieved by the algorithm is greater than optimal. In particular, 
if an algorithm 
claims an approximation ratio of $A$ and 
finds Hausdorff distance $H_{apr}$, then 
$H_{apr} \leq A\cdot H_{\textrm{opt}}$,
where $H_{\textrm{opt}}$ is the optimal Hausdorff distance. 
\fi
Goodrich {\it et al.}~\cite{goodrich1999approximate} present approximation 
algorithms for geometric pattern matching in multiple spaces under different 
types of motion, achieving
approximation ratios ranging from $2$ to $8+\epsilon$, for 
constant $\epsilon > 0$. 
\ifFull
Their algorithms are faster than 
the comparable
exact pattern-matching 
algorithms (e.g., $O(n^2m\log n)$, as compared to $O(n^2m^3\log^2{nm})$
for planar sets under Euclidean motion).
\fi
Cho and Mount~\cite{cho2008improved} show how to achieve improved
approximation ratios for such matching problems,
at the expense of making the analysis more complicated. 

Other algorithms give approximation ratios of $1+\epsilon$,
allowing the user to define the degree of certainty they want. 
Indyk {\it et al.}~\cite{indyk1999geometric} give a $(1+ \epsilon)$-approximation 
algorithm whose running time is defined in terms of both 
the number of points in the set as well as $\Delta$, which is defined as the the 
distance between the farthest and the closest pair of points. Cardoze and 
Schulman~\cite{cardoze1998pattern} offer a randomized $(1+\epsilon)$-approximation 
algorithm for $\R^d$ whose running time is also defined in terms of $\Delta$. 
These algorithms are fast when $\Delta$ is relatively small, which is true on 
average for many application areas, but these algorithms are much less efficient 
in domains where $\Delta$ is likely to be arbitrarily large.
\ifFull
Strictly speaking, these algorithms are 
not polynomial-time approximation schemes, 
because their running times depend on 
geometric properties of the input and not just on
its combinatorial complexity.
\fi

\subsection{Our Results}
In this paper, we present a family of simple algorithms for approximate 
oriented point-set pattern matching problems, that is, computational
geometry problems motivated by fingerprint matching. 

Each of our algorithms uses as a subroutine a \emph{base
algorithm} that selects certain points of the pattern, $P$, and ``pins'' them into certain positions with respect to the background, $B$.
This choice determines a transformed copy $P'$ of the whole point set $P$.
We then compute the directed Hausdorff distance for this transform
by querying the nearest neighbor in $B$ for each point of $P'$. 
To find 
nearest neighbors for  a suitably-defined metric on oriented points that combines straight-line distance with  rotation amounts, we adapt 
balanced box decomposition (BBD) trees~\cite{arya1998optimal} to oriented point sets,
which may be of independent interest. The general idea of this adaptation is to insert two copies of each point such that, for any query point, if we find its nearest neighbor
using the $L_1$/$L_2$-norm, we will either find the nearest neighbor based on $\mu_1$/$\mu_2$ or we will find one of its copies.
The output of the base algorithm is the transformed copy $P'$ that minimizes this distance.
We refer to our base algorithms as \emph{pin-and-query} methods.

These base algorithms are all simple and effective,
but their approximation factors are larger than $2$, whereas we seek $(1+\epsilon)$-approximation schemes for any constant $\epsilon > 0$.
To achieve such results,
our approximation schemes call the base algorithm twice. The first call determines an approximate scale of the solution. Next, our schemes apply a \emph{grid-refinement} strategy that expands the set of
background points by convolving it with a fine 
grid at that scale, in order to provide more candidate motions. Finally, they call
the base algorithm a second time on the expanded input.
This allows us to leverage the speed and simplicity of the base algorithms,
gaining greater accuracy while losing only a constant factor in our running times.

The resulting approximation algorithms run in the same asymptotic time bound as 
the base algorithm (with some dependence on $\epsilon$ in the constants) 
and achieve approximations that are a $(1+\epsilon)$ factor 
close to optimal, for any constant $\epsilon>0$. 
For instance, one of our approximation schemes, designed in this way,
guarantees a worst case running time of $O(n^2m\log n)$ for rigid
motions defined by translations and rotations. 
Thus, our approach results in polynomial-time approximation schemes (PTASs),
where their  running times depend only on combinatorial parameters.
Specifically, we give
the runtimes and approximations ratios for our algorithms 
in Table~\ref{tab:approx}.

\begin{table}[htb]
\begin{center}
  \begin{tabular}{ | l | c | c |}
    \hline
    {\bf Algorithm} & {\bf Running Time} & {\bf Approx. Ratio} \\ \hline\hline
    T & $O(nm\log n)$ & $1+\epsilon$  \\ \hline
    TR & $O(n^2m\log n)$ & $1+\epsilon$  \\ \hline
    TRS & $O(n^2m\log n)$ & $1+\epsilon$  \\ \hline
  \end{tabular}
\end{center}    
 \caption{Running times and approximation ratios for our approximation 
algorithms. \label{tab:approx}}
\end{table}

The primary challenge in the 
design of our algorithms is to come up with methods that 
achieve an approximation factor of $1+\epsilon$, for any 
small constant $\epsilon>0$, without resulting in a running time that is dependent
on a geometric parameter like $\Delta$.
The main idea that we use to overcome this challenge is for our base algorithms
in some cases
to use two different pinning schemes, one for large diameters and one for small diameters,
We show that one of these pinning schemes always finds a good match, so choosing
the best transformation found by either of them
allows us to avoid a dependence on geometric parameters in our running times.
As mentioned above, all of our base algorithms are simple, as are our 
$(1+\epsilon)$-approximation algorithms.
Moreover, proving each of our algorithms
achieves a good approximation ratio is also
simple, involving no more than ``high school'' geometry. 
Still, for the sake of our presentation, we postpone some proofs
and simple cases to appendices.

\section{Formal Problem Definition}
\label{sec:problem}
Let us formally define the 
\emph{oriented point-set pattern matching} problem. 
We define an \emph{oriented point set} in $\R^2$ to be a finite subset
of the set $O$ of all oriented points, defined as
\[
O = \bigl\{(x,y,a) \mid x,y,a \in \R, a \in [0,2\pi)\bigr\}.
\]
We consider three sets of transformations on oriented point sets, 
corresponding to the usual translations, rotations, and scalings on $\R^2$. 
In particular, we define the set of \emph{translations},
$\mathcal{T}$, as the set of functions $T_v:O\to O$ of the form
\[
T_v(x, y, a) = (x + v_x, y + v_y, a),
\]
where $v = (v_x, v_y) \in \R^2$ is referred to as the \emph{translation vector}.

Let $R_{p, \theta}$ be a rotation in $\R^2$ where $p$ and $\theta$
are the center and angle of rotation, respectively. We extend the action of $R_{p, \theta}$
from unoriented points to oriented points by defining
\[
R_{p, \theta}(x,y,a) = \bigl(R_{p, \theta}(x,y),(a+\theta)\bmod 2\pi\bigr),
\]
and we let $\mathcal{R}$ denote the set of rotation transformations from $O$ to $O$ defined in this way.

Finally, we define the set of \emph{scaling} operations on an oriented
point set. Each such operation $S_{p,s}$ is determined by a point $p = (x_p, y_p, a_p)$
at the center of the scaling and by a scale factor, $s$.
If a point $q$ is Euclidean distance $d$ away from $p$ before scaling, 
the distance between $q$ and $p$ should become $sd$ after scaling. 
In particular, this determines $S_{p,s}:O\to O$ to be the function
\[
S_{p,s}(x,y,a) = \bigl(x_p + s(x-x_p), y_p + s(y-y_p), a\bigr).
\]
We let $\mathcal{S}$ denote the set of scaling functions defined in this way. 

As in the unoriented point-set pattern matching problems, we use a directed
Hausdorff distance to measure how well a transformed patten 
set of points, $P$, matches a background set of points, $B$.
That is, we use
\[
h(P, B) = \max_{p \in P}\min_{q \in B}\mu(p, q),
\]
where $\mu(p,q)$ is a distance metric for oriented points in $\R^2$.
Our approach works for various types of metrics, $\mu$, for pairs of points,
but, for the sake of concreteness, we focus on two specific
distance measures for elements of $O$,
which are based on the $L_1$-norm and $L_2$-norm, respectively.
In particular,
for $(x_1, y_1, a_1), (x_2, y_2, a_2) \in  O$, let 
\begin{align*}
&\mu_1((x_1, y_1, a_1), (x_2, y_2, a_2)) = \\ 
&|x_1 - x_2| + |y_1 - y_2| + \min(|a_1 - a_2|, 2\pi - |a_1 - a_2|),
\end{align*}
and let
\begin{align*}
&\mu_2((x_1, y_1, a_1), (x_2, y_2, a_2)) = \\
&\sqrt{(x_1 - x_2)^2 + (y_1 - y_2)^2 + \min(|a_1 - a_2|, 2\pi - |a_1 - a_2|)^2}.
\end{align*}
Intuitively, one can interpret these distance metrics to be analogous to the 
$L_1$-norm and $L_2$-norm in a cylindrical 3-dimensional space where 
the third dimension wraps back around to $0$ at $2\pi$.
Thus, for $i\in\{1,2\}$, and
$B, P \subseteq O$, we
use the following directed Hausdorff distance:
\[
h_i(P, B) = \max_{p \in P} \min_{b \in B} \mu_i(p,b).
\]
Therefore, 
for some subset $\mathcal{E}$ of $\mathcal{T} \cup \mathcal{R} \cup \mathcal{S}$,
the \emph{oriented point-set pattern matching} problem is
to find a composition $E$ of one or more functions in $\mathcal{E}$ that 
minimizes $h_i(E(P), B)$. 

\section{Translations Only}
\label{translations}
In this section, we present our base algorithm and approximation
algorithm for approximately solving 
the oriented point-set pattern matching problem where we allow only translations. 
In this way, we present the basic template and data structures
that we will also use for the more interesting case of translations
and rotations ($\mathcal{T}\cup\mathcal{R}$).

Our methods for handling translations, rotations, and scaling is an
adaptation of our methods for 
$\mathcal{T}\cup\mathcal{R}$.

Given two subsets of $O$, $P$ and $B$, with $|P| = m$ and $|B| = n$, 
our goal here is to 
minimize $h_i(E(P), B)$ where $E$ is a transformation function in 
$\mathcal{T}$.

\subsection{Base Algorithm Under Translation Only}
\label{BaseAlgT}
Our base pin-and-query algorithm is as follows.

\begin{center}
\rule{\columnwidth}{2pt}
\textbf{Algorithm} BaseTranslate($P,B$):
\begin{algorithmic}
\STATE Choose some $p \in P$ arbitrarily.
\FOR {every $b \in B$}
\STATE \emph{Pin step:} 
Apply the translation, $T_v\in {\cal T}$, that takes $p$ to $b$. 
\FOR {every $q \in T_v(P)$}
\STATE \emph{Query step:}
Find a nearest-neighbor of $q$ in $B$ using the $\mu_i$ metric, and update
a candidate Hausdorff distance for $T_v$ accordingly.
\ENDFOR
\STATE \textbf{return} the smallest candidate Hausdorff distance
found as the smallest distance, $h_i(T_v(P), B)$.
\ENDFOR 
\end{algorithmic}
\vspace*{-4pt}
\rule{\columnwidth}{2pt}
\end{center}

This algorithm uses a similar approach to an algorithm of 
Goodrich {\it et al.}~\cite{goodrich1999approximate},
but it is, of course, 
different in how it computes nearest neighbors, since we must use
an oriented distance metric rather than unoriented distance metric.
One additional difference is that rather than find an exact nearest neighbor,
as described above,
we instead find an \emph{approximate} nearest neighbor of each point, $q$,
since we are ultimately 
designing an approximation algorithm anyway.
This allows us to achieve a faster running time.

In particular,
in the query step of the algorithm, for any point $q\in T_v(P)$,
we find a neighbor, $b \in B$, whose distance 
to $q$ is at most a $(1+\epsilon)$-factor more than the distance from $q$ to
its true nearest neighbor.
To achieve this result, we adapt the
balanced box-decomposition (BBD) tree
of Arya {\it et al.}~\cite{arya1998optimal} to oriented point sets.
Specifically,
we insert into the BBD tree the following set of $3n$ points in $\R^3$:
\begin{align*}
\bigl\{b,b',b'' \mid &b \in B, \\
&b'=(x_p, y_b, a_b + 2\pi), \\
&b''=(x_b, y_b, a_b - 2\pi)\bigr\}.
\end{align*}
This takes $O(n\log n)$ preprocessing and it allows the BBD tree 
to respond to nearest neighbor 
queries with an approximation factor of $(1+\epsilon)$ while using the 
$L_1$-norm or $L_2$-norm as the distance metric, since the BBD is effective
as an approximate nearest-neighbor data structure
for these metrics. 
Indeed, this is the main reason why we are using these norms as our 
concrete examples of $\mu_i$ metrics.
Each query takes $O(\log n)$ time, 
so computing a candidate Hausdorff distance for a given transformation takes
$O(m\log n)$ time. Therefore, since we perform the pin step 
over $n$ translations, the algorithm overall takes time $O(nm\log n)$.
To analyze the correctness of this algorithm, we start with a simple
observation that if we translate a point using a vector whose 
$L_i$-norm is $d$, 
then the distance between the translated point and its old position is $d$.

\begin{lemma}
\label{lem:tran-dist}
Let $(x, y, a)$ be an element of $O$. 
Consider a transformation $T_v$ in $\mathcal{T}$ 
where $v$ is a translation vector. 
Let $T_v(x, y, a) = (x', y', a)$. 
If the $L_i$-norm of $v$ is $\|v\|_i = d$, then $\mu_i\bigl((x,y,a), (x',y',a)\bigr) = d$,
where $i \in \{1,2\}$.
\end{lemma}
\begin{proof}
First consider the case where $i = 1$. By definition of $\mu_1$ and $T_v$, 
\begin{align*}
&\mu_1\bigl((x,y, a), (x',y', a)\bigr) \\
&= |x-x'| + |y-y'| + \min(|a-a|, 2\pi - |a-a|) \\
&= |v_x| + |v_y| \\
&= d.
\end{align*}
Now consider the case where $i = 2$:
\begin{align*}
&\mu_2\bigl((x,y,a), (x',y',a)\bigr) \\
&= \sqrt{(x-x')^2 + (y-y')^2 + \min(a-a, 2\pi - |a-a|)^2}\\
& = \sqrt{v_x^2 + v_y^2}\\
& = d.
\end{align*}
Thus, for either case, the lemma holds.
\end{proof}

\begin{theorem}
\label{thm:BaseAlgT}
Let $h_{\textrm{opt}}$ be $h_i(E(P), B)$ where $E$ is the translation
in $\mathcal{T}$ that attains the minimum of $h_i$. 
The algorithm above runs in time 
$O(nm\log n)$ and produces an approximation to $h_{\textrm{opt}}$ that is at most $(2+\epsilon)h_{\textrm{opt}}$, for either $h_1$ and $h_2$, for any fixed constant
$\epsilon>0$.
\end{theorem}

\begin{proof}
The $\epsilon$ term comes from the approximate nearest neighbor queries
using the BBD tree, and expanding $B$ to a set of size $3n$ by making a copy of each point in $B$ to have an angle that is $2\pi$ greater and 
less than its original value.
So it is sufficient to 
prove a $2$-approximation using exact nearest neighbor queries (while
building the BBD tree to return $(1+\epsilon/2)$-approximate nearest neighbors).
We prove this claim by a type of ``backwards'' analysis.
Let $E$ be a translation in $\mathcal{T}$ that attains the 
minimum of $h_i(E(P), B)$, and let $P'=E(P)$. 
Then every point $q \in P'$ is at most $h_{\textrm{opt}}$ from its closest 
background point in $B$. 
That is, for all $q$ in $P'$, there exists $b$ in $B$ such that 
$\mu_i(q,b) \le h_{\textrm{opt}}$. 
Let $b'\in B$ be the closest background point to the optimal 
position of $p$, where $p$ is the point we choose in the first step 
of the algorithm. Thus,
\[
\|(x_p, y_p) - (x_{b'}, y_{b'})\|_i \le \mu_i(p,b') \le h_{\textrm{opt}}.
\]
Apply the translation $T_v$ on $P'$ so that 
$p$ coincides with $b'$, which is equivalent to moving every point's position by 
$\|(x_p, y_p) - (x_{b'}, y_{b'})\|_i$. Hence, by 
Lemma~\ref{lem:tran-dist}, all points have moved at most $h_{\textrm{opt}}$.

As all points in the pattern started at most $h_{\textrm{opt}}$ away from a point in the background 
set and the translation $T_v$ moves all points at most $h_{\textrm{opt}}$, all points in 
$T_v(P')$ are at most $2h_{\textrm{opt}}$ from a point in the background set $B$. Since our algorithm checks 
$T_v$ as one of the translations in the second step of the algorithm, it will find a translation that is at least as good as $T_v$. Therefore, our algorithm 
guarantees an approximation of at most $2h_{\textrm{opt}}$,
for either $h_1$ and $h_2$.
\end{proof}

\subsection{A $(1+\epsilon)$-Approximation Algorithm Under Translations Only}
\label{EpsAlgT}
In this subsection, we utilize the algorithm from Section~\ref{BaseAlgT} to 
achieve a $(1+\epsilon)$-approximation when we only allow translations. 
Suppose, then, that we are
given two subsets of $O$, $P$ and $B$, with $|P| = m$ and $|B| = n$, 
and our goal is to 
minimize $h_i(E(P), B)$ over translations $E$ in $\mathcal{T}$.
Our algorithm is as follows:

\begin{enumerate}
\item Run the base algorithm, BaseTranslate($P,B$), from Section~\ref{BaseAlgT},
to obtain an approximation, $h_{apr} \leq A\cdot h_{\textrm{opt}}$.
\item For every $b \in B$, 
generate the point set 
\[
G_b = G\left(b, \frac{\epsilon\,h_{apr}}{A^2-A},
\left\lceil\frac{A^2-A}{\epsilon}\right\rceil\right)
\] 
for $h_1$ 
or \\
\[
G_b = G\left(b, \frac{\epsilon\sqrt{2}h_{apr}}{A^2-A}, 
\left\lceil\frac{A^2-A}{\epsilon\sqrt{2}}\right\rceil\right)
\]
for $h_2$.
Let $B'$ denote this expanded set of background points,
i.e., $B' = \bigcup_{b \in B}G_b$, and note that if $A$ is a constant,
then $|B'|$ is $O(n)$.
\item
Return the result from calling BaseTranslate($P,B'$), but restricting the
query step to finding nearest neighbors in $B$ rather than in $B'$.

\end{enumerate}

Intuitively, this algorithm
uses the base algorithm to give us a first approximation for the 
optimal solution. We then use this approximation to generate a 
larger set of points from which to derive transformations to test. 
We then use this point set again in the base algorithm 
when deciding which transformations to iterate over, 
while still using $B$ to compute nearest neighbors.
The first step of this algorithm runs in time $O(nm\log n)$, as we showed. 
The second step takes time proportional to the number of points which have to be generated, which is determined by $n$, our choice of the constant $\epsilon$, and the approximation ratio of our base algorithm $A$, which we proved is the constant $2+\epsilon$. 
The time needed to complete the second step is $O(n)$.
In the last step, we essentially call the base algorithm again on sets of size $m$
and $O(n)$, respectively; hence, this step requires
$O(nm\log n)$ time.

\begin{theorem}
\label{thm:EpsAlgT}
Let $h_{\textrm{opt}}$ be $h_i(E(P), B)$ where $E$ is the translation in 
$\mathcal{T}$ that attains the minimum of $h_i$, for $i \in \{1,2\}$. The algorithm above runs in time 
$O(nm\log n)$ and produces an approximation to $h_{\textrm{opt}}$ that is at most
$(1+\epsilon)h_{\textrm{opt}}$, for either $h_1$ and $h_2$.
\end{theorem}

\begin{proof}
Let $E$ be the translation in $\mathcal{T}$ 
that attains the minimum of $h_i(E(P), B)$. 
Let $P'$ be $E(P)$. Then every point $q \in P'$ is at most $h_{\textrm{opt}}$ from the closest 
background point in $B$.
By running the base algorithm the first time, we find $h_{apr}\leq A\cdot h_{\textrm{opt}}$, 
where $A$ is the approximation ratio of the base algorithm.
Now consider the point, $b'\in B$, 
that is the closest background to some pattern point $p \in P$. 
The square which encompasses $G_{b'}$ has a side length of $2h_{apr}$. 
This guarantees that $p$, which is at most $h_{\textrm{opt}}$ away from $b'$, lies within this square. 
As we saw from Lemma~\ref{lem:cube}, this means that $p$ is at most $\frac{\epsilon h_{apr}}{A^2-A}$ away from its nearest neighbor in $G_{b'}$. 
Thus, if a transformation defined by the nearest point in $B$ would move our pattern points at most $(A-1)h_{\textrm{opt}}$ from their optimal position, then using the nearest point in $G_{b'}$ to define our transformation will move our points at most $(A-1)\frac{\epsilon h_{apr}}{A^2-A} = \frac{\epsilon h_{apr}}{A} \leq \epsilon h_{\textrm{opt}}$.
Therefore, 
our algorithm gives a solution that is at most $(1+\epsilon)h_{\textrm{opt}}$
from optimal. 
\end{proof}

\section{Translations and Rotations}

In this section, we present our base algorithm and approximation
algorithm for approximately solving 
the oriented point-set pattern matching problem where we allow translations
and rotations. 
Given two subsets of $O$, $P$ and $B$, with $|P| = m$ and $|B| = n$, 
our goal here is to 
minimize $h_i(E(P), B)$ where $E$ is a composition of functions in 
$\mathcal{T}\cup \mathcal{R}$.
In the case of translations and rotations, we  
actually give two sets of algorithms---one set that works for point sets
with large diameter and one that works for point sets with small diameter.
Deciding which of these to use is based on a simple calculation (which we postpone
to the analysis below), which amounts to a normalization decision to determine how
much influence orientations have on matches versus coordinates.

\subsection{Base Algorithm Under Translation and Rotation with Large Diameter}
\label{BaseAlgTRlargeDiam}
In this subsection, we present an algorithm for solving 
the approximate oriented point-set pattern matching problem where 
we allow translations and rotations. 
This algorithm provides a good approximation ratio 
when the diameter of our pattern set is large.
Given two subsets $P$ and $B$ of $O$, with $|P| = m$ and $|B| = n$, we wish to minimize 
$h_i(E(P), B)$ over all compositions $E$ of one or more functions in 
$\mathcal{T} \cup \mathcal{R}$. 
Our algorithm is as follows
(see Figure~\ref{fig:BaseAlgTRlargeDiam}).

\begin{center}
\rule{\columnwidth}{2pt}

\textbf{Algorithm} BaseTranslateRotateLarge($P,B$):
\begin{algorithmic}
\STATE 
Find $p$ and $q$ 
in $P$ having the maximum value of $\|(x_p, y_p) - (x_q, y_q)\|_2$.
\FOR {every pair of points $b,b' \in B$}
\STATE \emph{Pin step:} 
Apply the translation, $T_v\in {\cal T}$, that takes $p$ to $b$,
and apply the rotation,
$R_{p,\theta}$, that makes $p$, $b'$, and $q$ collinear.
\STATE
Let $P'$ denote the transformed pattern set, $P$.
\FOR {every $q \in P'$}
\STATE \emph{Query step:}
Find a nearest-neighbor of $q$ in $B$ using the $\mu_i$ metric, and update
a candidate Hausdorff distance accordingly.
\ENDFOR
\STATE \textbf{return} the smallest candidate Hausdorff distance
found as the smallest distance, $h_i(R_{p,\theta}(T_v(P)), B)$.
\ENDFOR 
\end{algorithmic}
\vspace*{-4pt}
\rule{\columnwidth}{2pt}
\end{center}

\begin{figure}[hbt]
\centering
\includegraphics[width=.6\linewidth]{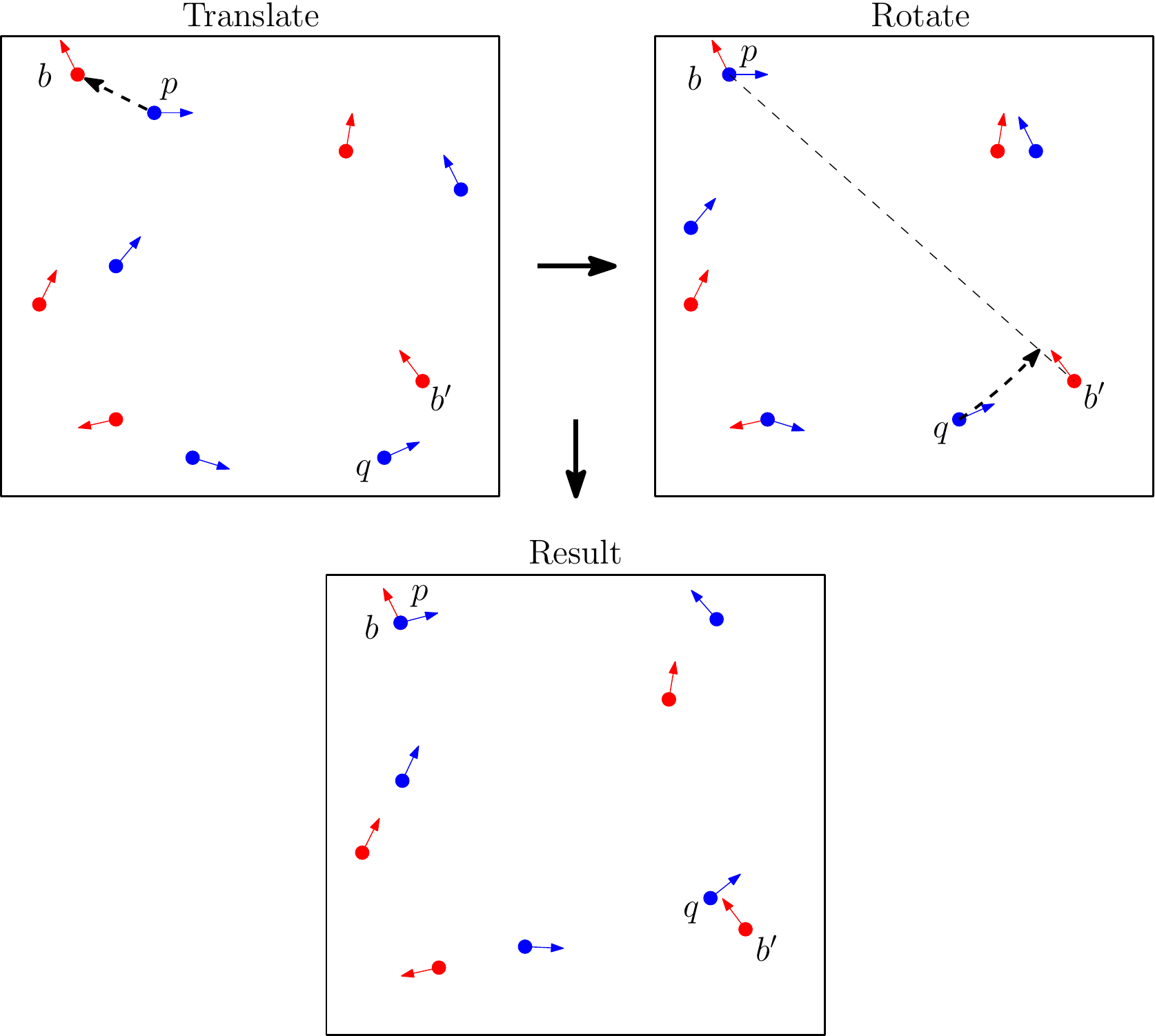}
\vspace*{-4pt}
\caption{Illustration of the translation and rotation steps of the base
approximation algorithm for translation and rotation in $O$ when diameter is
large.}
\label{fig:BaseAlgTRlargeDiam}
\end{figure}

\ifFull
This algorithm is somewhat similar to an algorithm by 
Goodrich et al.~\cite{goodrich1999approximate} for unoriented points,
but differs in how it computes distances and in the way
it computes nearest neighbors using BBD trees.
\fi

The points $p$ and 
$q$ can be found in $O(m \log m)$ time~\cite{preparatacomputational}. 
The pin step iterates over $O(n^2)$ translations and rotations, respectively,
and, for each one of these transformations,
we perform $m$ BBD queries, each of which takes $O(\log n)$ time.
Therefore, our total running time is 
$O(n^2m\log n)$.
Our analysis for this algorithm's approximation factor uses the following
simple lemma.

\begin{lemma}
\label{lem:rot-dist}
Let $P$ be a finite subset of $O$. Consider the rotation $R_{c, \theta}$ in $\mathcal{R}$. Let $q = (x_q, y_q, a_q)$ be the 
element in $P$ such that $\|(x_q,y_q)-(x_c, y_c)\|_2 = D$ is maximized.
For any $p = (x_p, y_p, a_p) \in P$, denote $R_{c, \theta}(x_p, y_p, a_p)$ as $p' = (x_{p'}, y_{p'}, a_{p'})$. Let $i \in \{1,2\}$. Then for 
all $p \in P$, $\mu_i(p, p') \leq \|(x_q, y_q) - (x_{q'}, y_{q'})\|_i + \pi
\|(x_q, y_q) - (x_{q'}, y_{q'})\|_2/(2D)$.
\end{lemma}

\begin{figure}[hbt]
\centering
\includegraphics[width=.6\linewidth]{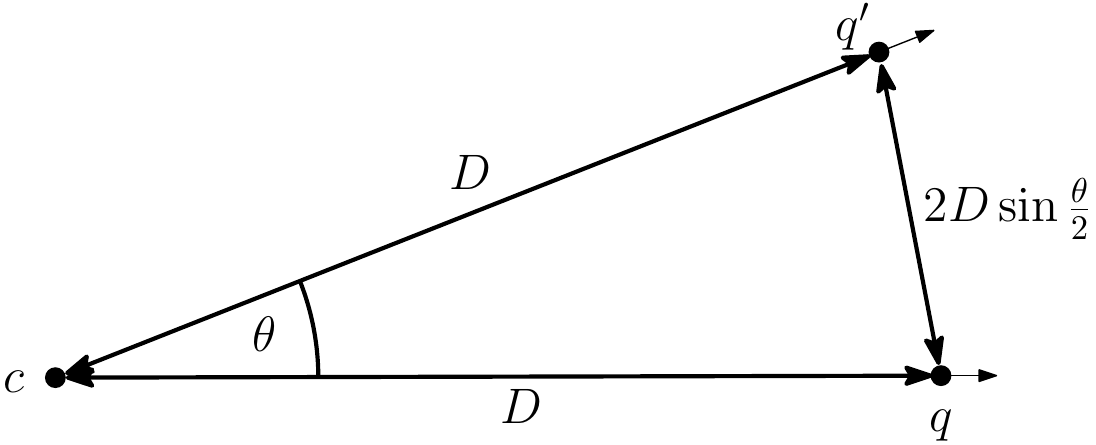}
\caption{The rotation of $q$ to $q'$ about $c$}
\label{fig:rot-dist}
\end{figure}

\begin{proof}
After applying the rotation $R_{c, \theta}$, we know $q$ has moved at least as far than any other point
because it is the farthest from the center of rotation. Without loss of generality, $0 \le \theta \le \pi$. 
Then it is easily verifiable 
that $\theta/\pi \le \sin(\theta/2)$. As $2D\sin(\theta/2)$ is 
the Euclidean distance $q$ moves under $R_{c, \theta}$, it follows that
\[
\frac{2D\theta}{\pi} \le 2D\sin(\theta/2) = \|(x_q, y_q) - (x_{q'}, y_{q'})\|_2.
\]
This scenario is illustrated in Figure~\ref{fig:rot-dist}.
Thus, $\theta \le (\pi\|(x_q, y_q) - (x_{q'}, y_{q'})\|_2)/(2D)$, which implies that $R_{c,\theta}$ 
moves the position of $q$ by at most $\|(x_q, y_q) - (x_{q'}, y_{q'})\|_i$ and changes the orientation of $q$ 
by at most $\pi\|(x_q, y_q) - (x_{q'}, y_{q'})\|_2/(2D)$. Therefore, because $q$ moves farther than any other 
point in $P$, any point $p \in P$ has moved a distance of at most $\|(x_q, y_q) - (x_{q'}, y_{q'})\|_i + \pi\|
(x_q, y_q) - (x_{q'}, y_{q'})\|_2/(2D)$ with respect to the distance function $\mu_i$.
\end{proof}

\begin{theorem}
\label{thm:BaseAlgTRlargeDiam}
Let $h_{\textrm{opt}}$ be $h_i(E(P), B)$ where $E$ is the composition of functions in $
\mathcal{T} \cup \mathcal{R}$ that attains the minimum of $h_i$, for 
$i\in\{1,2\}$. The 
algorithm above runs in time $O(n^2m\log n)$ and produces an approximation to 
$h_{\textrm{opt}}$ that is at most 
$(A_1 + \epsilon)h_{\textrm{opt}}$ for $h_1$ and 
at most $(A_2 + \epsilon)h_{\textrm{opt}}$ for $h_2$, where $\epsilon>$ is a fixed
constant,
$A_1=6 + \sqrt{2}\pi/D$, and 
$A_2=2 + \sqrt{2}(2 + \pi/D)$.
\end{theorem}

\begin{proof}
The additional $\epsilon$ terms come entirely from using 
approximate nearest neighbor queries (defining BBD trees so they return
$(1+\epsilon/A_i)$-approximate nearest neighbors, for $i\in\{1,2\}$).
So it is sufficient for us to prove approximation bounds
that are $A_i\cdot h_{\textrm{opt}}$.

The first step is argued similarly to that of 
the proof of Theorem~\ref{thm:BaseAlgT}. 
Let $E$ be the composition of functions in $\mathcal{T} \cup \mathcal{R}$ 
that attains the minimum of $h(E(P), B)$ and let $P'$ be $E(P)$. Then for all $p$ in $P'$, there exists $b$ in $B$ such that $\mu_i(p,b) \le h_{\textrm{opt}}$. 
Let $p',q' \in B$ be the closest background points to optimal positions of $p$ and $q$  respectively, where $p$ and $q$ are the diametric points we choose in the first step of the algorithm. Thus, 
\[
\|(x_p, y_p) - (x_{p'}, y_{p'})\|_i \le \mu_i(p,p') \le h_{\textrm{opt}}.
\]
Apply the translation $T_v$ on $P'$ so that 
$p$ coincides with $p'$, which is equivalent to moving every point 
$\|(x_p, y_p) - (x_{p'}, y_{p'})\|_i$ with respect to position.
Lemma~\ref{lem:tran-dist}, then, implies that all points have moved at most $h_{\textrm{opt}}$.

Next, apply the rotation $R_{p, \theta}$ to $P'$ that makes $p, q$, and $q'$  
co-linear. With respect to position, $q$ moves at 
most a Euclidean distance of $2D\sin(\theta/2)$ away from $q'$ where $D$ is the Euclidean 
distance between $p$ and $q$. As all points were already at most $2h_{\textrm{opt}}$ 
away from their original background point in $B$, this implies that $2D
\sin(\theta/2) \le 2\sqrt{2}h_{\textrm{opt}}$. Thus, $\|(x_q, y_q) - (x_{q'}, y_{q'})\|_2$ is at most 
$2\sqrt{2}h_{\textrm{opt}}$. Then by Lemma~\ref{lem:rot-dist}, as $q$ is the furthest
point from $p$, the rotation moves all points at most $2\sqrt{2}h_{\textrm{opt}} + \sqrt{2}\pi h_{\textrm{opt}}/D$ with respect to $h_2$ and at most $4h_{\textrm{opt}} + \sqrt{2}\pi h_{\textrm{opt}}/D$ for $h_1$.

Since each point in the pattern set started out at most $h_{\textrm{opt}}$ away from a 
point in the background set, we combine this with the translation and rotation movements to find that every point ends up at most $(6 + \sqrt{2}\pi/D)h_{\textrm{opt}}$ away from a background point for $h_1$ and at most $(2 + \sqrt{2}(2 + \pi/D))h_{\textrm{opt}}$ away from a background point for $h_2$. As our algorithm checks this 
combination of $T_v$ and $R_{p, \theta}$, our algorithm guarantees at least this solution. Note that we assume $p'$ and $q'$ are not the same point. However if this is the case, then we know that $D \leq 2h_{\textrm{opt}}$ thus when we translate $p$ to $p'$ every point is within $(\sqrt{5}+2\pi/D)h_{\textrm{opt}}$ of $p'$, which is a better approximation than the case where $p' \neq q'$ under our assumption that $D$ is large. 
\end{proof}

\subsection{Grid Refinement}
In this subsection, we describe our grid refinement process,
which allows us to use a base algorithm 
to obtain an approximation ratio of $1+\epsilon$. 
To achieve this result, we take advantage of an important property of the fact
that we are approximating a Hausdorff distance by a pin-and-query algorithm.
Our base algorithm approximates $h_{\textrm{opt}}$
by pinning a reference pattern point, $p$, to a background point, $b$.
Reasoning backwards, if we have a pattern in an optimal position, where every pattern point, $p$,
is at distance $d\le h_{\textrm{opt}}$ from its associated nearest neighbor in the
background, then 
one of the transformations tested by the base pin-and-query algorithm moves each pattern 
point by a distance of at most $(A_i-1)d$ away from this optimal location when it performs its pinning
operation.

Suppose we could 
define a constant-sized ``cloud'' of points with respect
to each
background point, such that one of these points is guaranteed to be very close
to the optimal pinning location, much closer than the distance $d$ from the above argument.

Then, if we use these cloud points to define the transformations checked by the base algorithm,
one of these transformations will move each point from its optimal position by a much smaller distance. 

To aid us in defining such a cloud of points,
consider the set of points 
$G(p,l,k) \subset \R^2$ (where $p = (x_p, y_p)$ is some point in $\R^2
$, $l$ is some positive real value, and $k$ is some positive integer) defined 
by the following formula:
\begin{align*}
G(p,l,k) &= \bigl\{q \in \R^2 \mid \\ 
					q &= (x_p + il, y_p + jl), i,j \in \mathbb{Z}, -k \leq
i,j \leq k\bigr\}.
\end{align*}
Then $G(p,l,k)$ is a grid of $(2k+1)^2
$ points within a square of side length $2kl$ centered at $p$, where the 
coordinates of each point are offset from the coordinates of $p$ by a multiple 
of $l$. An example is shown in Figure~\ref{fig:cube}.

\begin{figure}[hbt]
\centering
\includegraphics[width=0.3\linewidth]{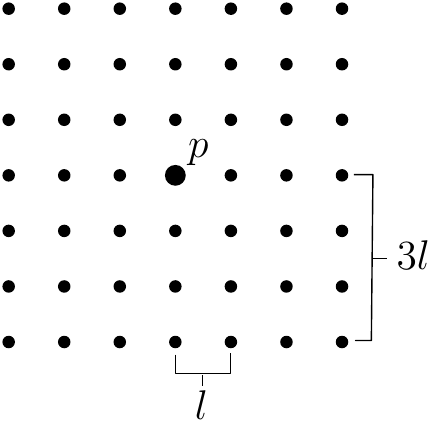}
\vspace*{-8pt}
\caption{An example of $G(p,l,3)$.}
\label{fig:cube}
\end{figure}

Now consider any point $q$ whose Euclidean distance is no more than 
$kl$ from $p$. This small distance forces point $q$ to lie within the square convex hull of 
$G(p,l,k)$. This implies that there is a point of $G(p,l,k)$ that is even closer to $q$:

\begin{lemma}
\label{lem:cube}
Let $i \in \{1,2\}$. Given two points $p,q \in \mathbb{R}^2$, if $\|p -q\|_i \leq kl$, then $\|q - s\|_1 \leq l$ and 
$\|q - s\|_2 \leq l/\sqrt{2}$, where $s$ is $q$'s closest neighbor in $G(p,l,k)
$.
\end{lemma}

\begin{proof}
Because $\|p-q\|_i \leq kl$, we know that $q$ exists within the square of side 
length $2kl$ which encompasses $G(p,l,k)$ (which we will refer to as $G$ for 
the remainder of this proof). This square can be divided into $(2k)^2$ 
non-overlapping squares of side length $l$. It is easy to see that the vertices of 
these squares are all points in $G$ and that $q$ exists within (or on the edge of) at least 
one of these squares. The point inside of a square that maximizes the 
distance to the square's closest vertex is the exact center of the square. 
If the side length is $l$, simple geometry shows us that at this point, the 
distance to any vertex is $l$ with respect to the $L_1$-norm and $l/\sqrt{2}$ with
respect to the $L_2$-norm. Thus, because $q$ exists within 
a square of side length $l$ whose vertices are points in $G$, the furthest 
that $q$ can be from its nearest neighbor in $G$ is $l$ for the $L_1$-norm 
and $l/\sqrt{2}$ for the $L_2$-norm.
\end{proof}

\subsection{A $(1+\epsilon)$-Approximation Algorithm Under Translation and Rotation with Large Diameter}
\label{EpsAlgTRlargeDiam}
Here, 
achieve a $(1+\epsilon)$-approximation ratio 
when we allow translations and rotations. 
Again, given two subsets of $O$, $P$ and $B$, with $|P| = m$ and $|B| = n$, our goal is to 
minimize $h_i(E(P), B)$ over all compositions $E$ of one or more functions in 
$\mathcal{T} \cup \mathcal{R}$.
We perform the following steps.
\begin{enumerate}
\item Run algorithm, BaseTranslateRotateLarge($P,B$),
from Section~\ref{BaseAlgTRlargeDiam} to obtain 
an approximation $h_{apr} \leq A\cdot h_{\textrm{opt}}$, where $A=A_1+\epsilon$ or
$A=A_2+\epsilon$, for a constant $\epsilon>0$.
\item For every $b \in B$, generate the grid of points $G_b = G(b, \frac{h_{apr}\epsilon}{A^2-A}, \lceil\frac{A^2-A}{\epsilon}\rceil)$ for $h_1$ or the grid $G_b = G(b, \frac{\sqrt{2}h_{apr}\epsilon}{A^2-A}, \lceil\frac{A^2-A}{\sqrt{2}\epsilon}\rceil)$ for $h_2$.
Let $B'$ denote the resulting point set, which is of size $O(A^4n)$, 
i.e., $|B'|$ is $O(n)$ when $A$ is a constant.
\item Run algorithm, BaseTranslateRotateLarge($P,B'$),
except use the original set, $B$, for nearest-neighbor queries in the query step.

\end{enumerate}

\ifFull
This algorithm uses the base algorithm to give us an 
indication of what the optimal solution might be, and then 
we use this approximation to generate a larger set of points from 
which to derive transformations to test. 
We then use this point set in the base algorithm when deciding which transformations to iterate over, while still using $B$ to compute nearest neighbors.
\fi
It is easy to see that this simple algorithm runs in 
$O(A^8n^2m\log n)$, 
which is $O(n^2m\log n)$ when $A$ is a constant (i.e., when the points in $P$
have a large enough diameter).

\begin{theorem}
\label{thm:EpsAlgTRlargeDiam}
Let $h_{\textrm{opt}}$ be $h_i(E(P), B)$ where $E$ is the composition of functions in 
$\mathcal{T} \cup \mathcal{R}$ that attains the minimum of $h_i$. The algorithm above runs in time 
$O(A^8n^2m\log n)$ and produces an approximation to $h_{\textrm{opt}}$ that is at most $(1+\epsilon)h_{\textrm{opt}}$ for
both $h_1$ and $h_2$.
\end{theorem}

\begin{proof}
Let $E$ be the composition of functions in $\mathcal{T} \cup \mathcal{R}$ that attains the minimum of $h(E(P), B)$. 
Let $P'$ be $E(P)$. Then every point $q \in P'$ is at most $h_{\textrm{opt}}$ from the closest 
background point in $B$.
By running the base algorithm, we find $h_{apr}\leq A\cdot h_{\textrm{opt}}$,
where $A$ is the approximation ratio of the base algorithm. 
Now consider the point $b'\in B$ which is the closest background to some pattern point $p \in P$. 
The square which encompasses $G_{b'}$ has a side length of $2h_{apr}$. 
This guarantees that $p$, which is at most $h_{\textrm{opt}}$ away from $b'$, lies within this square. 
As we saw from Lemma~\ref{lem:cube}, this means that $p$ is at most $\frac{\epsilon h_{apr}}{A^2-A}$ away from its nearest neighbor in $G_{b'}$. 
Thus, if a transformation defined by the nearest points in $B$ would move our pattern points at most $(A-1)h_{\textrm{opt}}$ from their optimal position, then using the nearest points in $G_{b'}$ to define our transformation will move our points at most $(A-1)\frac{\epsilon h_{apr}}{A^2-A} = \frac{\epsilon h_{apr}}{A} \leq \epsilon h_{\textrm{opt}}$. 
Thus, the modified algorithm gives a solution that is at most $(1+\epsilon)h_{\textrm{opt}}$. 
\end{proof}

\subsection{Base Algorithm Under Translation and Rotation with Small Diameter}
\label{BaseAlgTRsmallDiam}
In this subsection, we present an alternative 
algorithm for solving the approximate oriented point-set pattern matching problem where we allow translations and rotations. Compared to the algorithm given in Section~\ref{BaseAlgTRlargeDiam}, this algorithm instead provides a good approximation ratio when the diameter of our pattern set is small. 
Once again, given two subsets of $O$, $P$ and $B$, with $|P| = m$ and $|B| = n$, we wish to minimize $h_i(E(P), B)$ over all compositions $E$ of one or more functions in 
$\mathcal{T} \cup \mathcal{R}$. 
We perform the following algorithm
(see Figure~\ref{fig:BaseAlgTRsmallDiam}).

\begin{center}
\rule{\columnwidth}{2pt}
\textbf{Algorithm} BaseTranslateRotateSmall($P,B$):
\begin{algorithmic}
\STATE 
Choose some $p \in P$ arbitrarily.
\FOR {every points $b \in B$}
\STATE \emph{Pin step:} 
Apply the translation, $T_v\in {\cal T}$, that takes $p$ to $b$,
and then
apply the rotation,
$R_{p,\theta}$, that makes $p$ and $b$ have the same orientation.
\STATE
Let $P'$ denote the transformed pattern set, $P$.
\FOR {every $q \in P'$}
\STATE \emph{Query step:}
Find a nearest-neighbor of $q$ in $B$ using the $\mu_i$ metric, and update
a candidate Hausdorff distance accordingly.
\ENDFOR
\STATE \textbf{return} the smallest candidate Hausdorff distance
found as the smallest distance, $h_i(R_{p,\theta}(T_v(P)), B)$.
\ENDFOR 
\end{algorithmic}
\vspace*{-4pt}
\rule{\columnwidth}{2pt}
\end{center}

\begin{figure}[hbt]
\centering
\includegraphics[width=.6\linewidth]{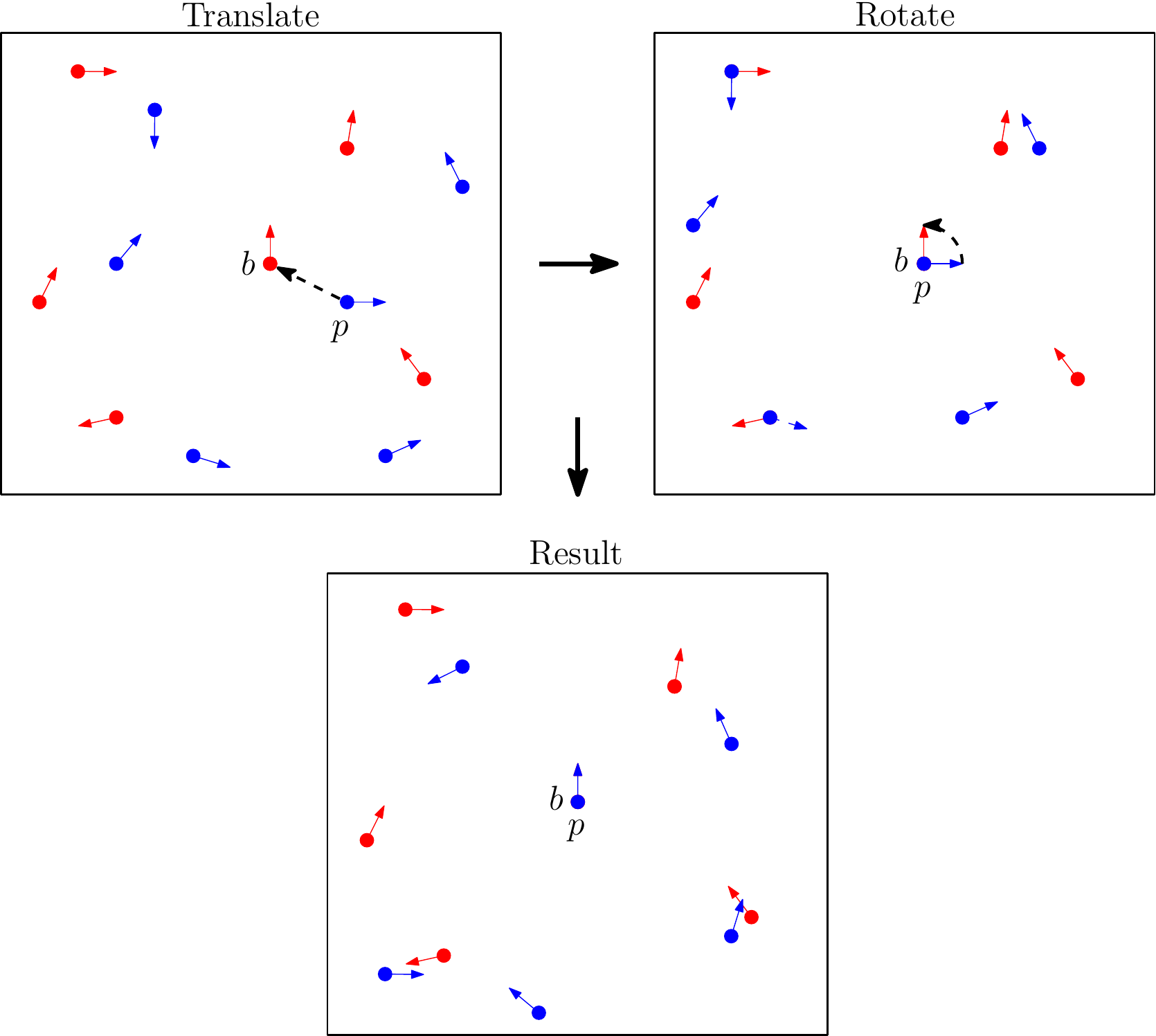}
\vspace*{-4pt}
\caption{Illustration of the translation and rotation steps of the base
approximation algorithm for translation and rotation in $O$ when diameter is
small.}
\label{fig:BaseAlgTRsmallDiam}
\end{figure}

\begin{theorem}
\label{thm:BaseAlgTRsmallDiam}
Let $h_{\textrm{opt}}$ be $h_i(E(P), B)$ where $E$ is the composition of functions in $
\mathcal{T} \cup \mathcal{R}$ that attains the minimum of $h_i$. The 
algorithm above runs in time $O(nm\log n)$ and produces an approximation to 
$h_{\textrm{opt}}$ that is at most 
$(A_i + \epsilon)h_{\textrm{opt}}$ for $h_i$, where $i=\{1,2\}$,
$\epsilon>0$ is a fixed constant,
$A_1 = 2 + \sqrt{2}D$,
and $A_2 = 2 + D$.
\end{theorem}

\begin{proof}
The additional $\epsilon$ terms come entirely from using 
approximate nearest neighbor queries, so it is sufficient 
to prove approximations which do not include the $\epsilon$ term using 
exact nearest neighbor queries (defining the BBD tree so that it returns
points that are $(1+\epsilon/A_i)$-approximate nearest neighbors). 
Particularly, we will prove a bound of $(2 + \sqrt{2}D)h_{\textrm{opt}}$ for $h_1$ and a bound of $(2 + D)h_{\textrm{opt}}$ for $h_2$.

Let $E$ be the composition of functions in $\mathcal{T} \cup \mathcal{R}$ 
that attains the minimum of $h(E(P), B)$. Let $P'$ be $E(P)$. Then every point 
$p \in P'$ is at most $h_{\textrm{opt}}$ from the closest background point in $B$. That 
is, for all $p$ in $P'$, there exists $b$ in $B$ such that $\mu_i(p,b) \le h_{\textrm{opt}}$. 
Let $p' \in B$ be the closest background point to the optimal position of $p$ where $p$ is the point we chose in the first step of the algorithm. Thus, 
\[
\mu_i(p,p') \le h_{\textrm{opt}}.
\]
Apply the translation $T_v$ and rotation $R_{p, \theta}$ on $P'$ so that 
$p$ coincides with $p'$ and both points have the same orientation. It is easy to see that $p$ has moved from its optimal position by exactly $\mu_i(p, p') \leq h_{\textrm{opt}}$. Using Lemma~\ref{lem:tran-dist} and the fact that a rotation on $P$ causes the orientation of each point in $P$ to change by the same amount, we find that every point $q \in P$ has moved at most $\mu_i(p, p') + d$ from its original position, where $d$ is the change in the position of $q$ caused by the rotation.

We know that the angle rotated, $\theta$, must be less than $h_{\textrm{opt}}$ and, without loss of generality, we assume $0 \leq \theta \leq \pi$. Therefore it is easily verifiable that $\sin(\theta/2) \leq \theta/2$. If $D$ is the diameter of $P$, then regardless of our choice of $p$, each point in $P$ is displaced at most $2D\sin(\theta/2)$ by the rotation. Thus each point is displaced at most $D\theta \leq Dh_{\textrm{opt}}$.

Since each point in the pattern set started out at most $h_{\textrm{opt}}$ away from a 
point in the background set, we combine this with the translation and rotation movements to find that every point ends up at most $(2 + \sqrt{2}D)h_{\textrm{opt}}$ away from a background point for $h_1$ and at most $(2 + D)h_{\textrm{opt}}$ away from a background point for $h_2$. As our algorithm checks this combination of $T_v$ and $R_{p, \theta}$, our algorithm guarantees at least this solution.
\end{proof}

\subsection{A $(1+\epsilon)$-Approximation Algorithm Under 
Translation and Rotation with Small Diameter}
\label{EpsAlgTRsmallDiam}
In this subsection, we utilize the algorithm from 
Section~\ref{BaseAlgTRsmallDiam} to achieve a $(1+\epsilon)$-approximation ratio when we allow translations and rotations. Again, given two subsets of $O$, $P$ and $B$, with $|P| = m$ and $|B| = n$, our goal is to 
minimize $h_i(E(P), B)$ over all compositions $E$ of one or more functions in 
$\mathcal{T} \cup \mathcal{R}$.
We begin by describing another type of grid refinement we use in this case.

In particular, let us consider a set of points 
$C(p,k) \subset O$ where $p = (x_p, y_p, a_p)$ is some point in $O$ 
and $k$ is some positive integer. We define 
the set in the following way (see Figure~\ref{fig:circle}):
\begin{align*}
C(p,k) &= \{q \in O \vert \\
q &= (x_p, y_p, a+2\pi
i/k \mod 2\pi),  i \in \mathbb{Z}, 1 \leq i \leq k\}.
\end{align*}
 
\begin{figure}[hbt]
\centering
\includegraphics[width=0.35\linewidth]{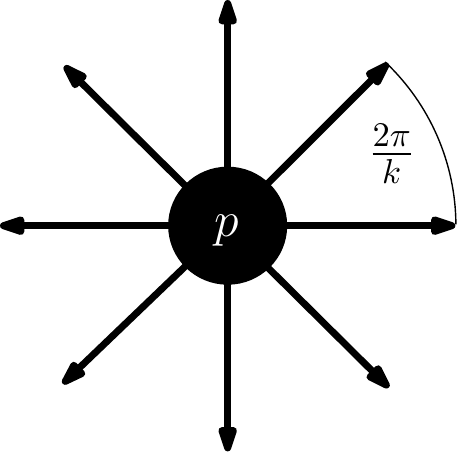}
\caption{An example of $C(p,8)$.}
\label{fig:circle}
\end{figure}

From this definition, we can see that $C(p,k)$ is a set of points
that share the same position as $p$ but have different orientations that are equally spaced out, with each point's orientation being an angle of $\frac{2\pi}{k}$ away from the previous point. 
Therefore, it is easy to see that, for any point $q \in O$, 
there is a point in $C(p,k)$ whose orientation is at 
most an angle of $\frac{\pi}{k}$ away from the orientation of $q$. 
Given this definition, our algorithm is as follows.

\begin{enumerate}
\item Run algorithm, BaseTranslateRotateSmall($P,B$),
from Section~\ref{BaseAlgTRsmallDiam}, 
to obtain $h_{apr} \leq A\cdot h_{\textrm{opt}}$.
\item For every $b \in B$, generate the point set
\[
G_b = G\left(b, \frac{h_{apr}\epsilon}{2(A^2-A)}, \left\lceil\frac{2(A^2-A)}{\epsilon}\right\rceil\right)
\]
for $h_1$ or
\[
G_b = G\left(b, \frac{h_{apr}\epsilon}{A^2-A}, \left\lceil\frac{A^2-A}{\epsilon}\right\rceil\right)
\]
for $h_2$.
Let $B'$ denote the resulting set of points, i.e.,
$B' = \bigcup_{b \in B}G_b$.
\item 
For every $b' \in B'$,
generate the point set
\[
C_{b'} = C\left(b', \frac{2(A^2-A)}{\pi h_{apr}\epsilon}\right)
\]
for $h_1$ or
\[
C_{b'} = C\left(b', \frac{\sqrt{2}(A^2-A)}{\pi h_{apr}\epsilon}\right)
\]
for $h_2$.
Let $B''$ denote the resulting set of points.
\item Run algorithm, BaseTranslateRotateSmall($P,B''$),
but continue to use the points in $B$ for nearest-neighbor queries.

\end{enumerate}

Intuitively, this algorithm uses the base algorithm to give us an indication of what the optimal solution might be. 
We then use this approximation to generate a larger set of points from which to derive transformations to test, but this time we also generate a number of different orientations for those points as well. We then use this point set in the base algorithm when deciding which transformations to iterate over, while still using $B$ to compute nearest neighbors.

The first step of this algorithm runs in time $O(nm\log n)$, as we showed. 
The second step takes time proportional to the number of points which have to be generated, which is determined by $n$, our choice of the constant $\epsilon$, 
and the approximation ratio, $A$, of our base algorithm. 
The time needed to complete the second step is $O(A^4n)$. 
The third step generates even more points based on points generated in step two, 
which increases the size of $B''$ to be $O(A^6n)$. 
The last step runs in time $O(A^6nm\log n)$, which is also the running time 
for the full algorithm. 

\begin{theorem}
\label{thm:EpsAlgTRsmallDiam}
Let $h_{\textrm{opt}}$ be $h_i(E(P), B)$ where $E$ is the composition of functions in 
$\mathcal{T} \cup \mathcal{R}$ that attains the minimum of $h_i$. The algorithm above runs in time 
$O(A^6nm\log n)$ and produces an approximation to $h_{\textrm{opt}}$ that is at most $(1+\epsilon)h_{\textrm{opt}}$ for
both $h_1$ and $h_2$.
\end{theorem}

\begin{proof}
Let $E$ be the composition of functions in $\mathcal{T} \cup \mathcal{R}$ that attains the minimum of $h_i(E(P), B)$. 
Let $P'$ be $E(P)$. Then every point $q \in P'$ is at most $h_{\textrm{opt}}$ from the closest 
background point in $B$.
By running the base algorithm, we find $h_{apr}\leq A\cdot h_{\textrm{opt}}$ where $A$ is the approximation ratio of the base algorithm. 
Now consider the point $b'\in B$ which is the closest background to some pattern point $p \in P$. 
The square which encompasses $G_{b'}$ has a side length of $2h_{apr}$. 
This guarantees that $p$, which is at most $h_{\textrm{opt}}$ away from $b'$, lies within this square. 
As we saw from Lemma~\ref{lem:cube}, this means that $p$ is at most $\frac{\epsilon h_{apr}}{2(A^2-A)}$ away from its nearest neighbor $g$ in $G_{b'}$ with respect to the $L1$-norm, and at most $\frac{\epsilon h_{apr}}{\sqrt{2}(A^2-A)}$ with respect to the $L2$-norm. 
For this point, $g$, there are a number of points in $C_g$ which are at the same position but with different orientation. 
For some point $c$ in $C_g$, the orientation of point $p$ is within an angle of at most $\frac{h_{apr}\epsilon}{2(A^2-A)}$ for $h_1$ and at most $\frac{h_{apr}\epsilon}{\sqrt{2}(A^2-A)}$ for $h_2$. 
If we combine together the maximum difference in position between $p$ and $c$, and the maximum difference in orientation between $p$ and $c$, then we see that for both $\mu_1$ and $\mu_2$, the distance between $p$ and $c$ is at most $\frac{h_{apr}\epsilon}{A^2-A}$. 
Thus, if a transformation defined by the nearest point in $B$ would move our pattern points at most $(A-1)h_{\textrm{opt}}$ from their optimal position, then using the nearest point in $C_{g}$ to define our transformation will move our points at most $(A-1)\frac{\epsilon h_{apr}}{A^2-A} = \frac{\epsilon h_{apr}}{A} \leq \epsilon h_{\textrm{opt}}$. 
Thus, the modified algorithm gives a solution that is at most $(1+\epsilon)h_{\textrm{opt}}$. 
\end{proof}

\subsection{Combining the Algorithms for Large and Small Diameters}

For the two cases above, we provided two base algorithms that each had a corresponding $(1+\epsilon)$-approximation algorithm. 
As mentioned above,
we classified the two by whether the algorithm 
achieved a good approximation when the 
diameter of the pattern set was large or small.
This is because the large diameter base algorithm has 
an approximation ratio with terms that are inversely proportional to 
the diameter, and the small diameter base algorithm has
an approximation ratio 
with terms that are directly proportional to the diameter. 

Both of the resulting $(1+\epsilon)$-approximation algorithms 
have running times which are affected by the approximation ratio of their base algorithm, meaning their run times are dependent upon the diameter of the pattern set. 
We can easily see, however, that the approximation ratios of the 
large and small diameter base algorithms intersect at 
some fixed constant diameter, $D^*$. 
For $h_1$, if we compare the expressions $6 + \sqrt{2}\pi/D$ and $2+\sqrt{2}D$, we find that the two expressions are equal at $D^*=\sqrt{2} + \sqrt{2+\pi} \approx 3.68$. 
For $h_2$, we compare $2 + \sqrt{2}(2 + \pi/D)$ and $2+D$ to find that they are equal at $D^* = \sqrt{2} + \sqrt{2+\sqrt{2}\pi} \approx 3.95$. 
For diameters larger than $D^*$, the approximation ratio of the large diameter algorithm is smaller than at $D^*$, and for diameters smaller than $D^*$, the approximation ratio of the small diameter algorithm is smaller than at $D^*$. 
Thus, if we choose to use the small diameter algorithms when the diameter is less than $D^*$ and we use the large diameter algorithms when the diameter is greater or equal to $D^*$, we ensure that the approximation ratio is no more than the constant
value that depends on the constant $D^*$.
Thus, based on the diameter of the pattern set, we can decide \textit{a priori} if 
we should use our algorithms for large 
diameters or small diameters and just go with that set of
algorithms.
This implies that we are guaranteed that our approximation factor, $A$, in our
base algorithm is always bounded above by a constant; hence, our running time for the 
translation-and-rotation case is $O(n^2 m\log n)$.

\section{Translation, Rotation, and Scaling}
In this section, we show how to adapt our algorithm for translations and rotations
so that it works for translations, rotations, and scaling.
The running times are the same as for the translation-and-rotation cases.

\subsection{Base Algorithm Under Translation, Rotation and Scaling with Large Diameter}
\label{BaseAlgTRSlargeDiam}
In this section we present an algorithm for solving the approximate oriented point-set pattern matching problem where we allow translations, rotations and scaling. This algorithm is an extension of the algorithm from Section~\ref{BaseAlgTRlargeDiam} and similarly provides a good approximation ratio when the diameter of our pattern set is large. 
Given two subsets $P$ and $B$ of $O$, with $|P| = m$ and $|B| = n$, we wish to minimize 
$h_i(E(P), B)$ over all compositions $E$ of one or more functions in 
$\mathcal{T} \cup \mathcal{R} \cup \mathcal{S}$. 
We perform the following algorithm:

\begin{center}
\rule{\columnwidth}{2pt}

\textbf{Algorithm} BaseTranslateRotateScaleLarge($P,B$):
\begin{algorithmic}
\STATE 
Find $p$ and $q$ 
in $P$ having the maximum value of $\|(x_p, y_p) - (x_q, y_q)\|_2$.
\FOR {every pair of points $b,b' \in B$}
\STATE \emph{Pin step:} 
Apply the translation, $T_v\in {\cal T}$, that takes $p$ to $b$,
and apply the rotation,
$R_{p,\theta}$, that makes $p$, $b'$, and $q$ collinear.
Then apply the scaling, $S_{p,s}$, that makes $q$ and $b'$ share the same position.
\STATE
Let $P'$ denote the transformed pattern set, $P$.
\FOR {every $q \in P'$}
\STATE \emph{Query step:}
Find a nearest-neighbor of $q$ in $B$ using the $\mu_i$ metric, and update
a candidate Hausdorff distance accordingly.
\ENDFOR
\STATE \textbf{return} the smallest candidate Hausdorff distance
found as the smallest Hausdorff distance, $h_i(S_{p,s}(R_{p,\theta}(T_v(P))), B)$.
\ENDFOR 
\end{algorithmic}
\vspace*{-4pt}
\rule{\columnwidth}{2pt}
\end{center}

This algorithm extends the algorithm presented in Section~\ref{BaseAlgTRlargeDiam} so that after translating and rotating, we also scale the point set such that, after scaling, $p$ and $b$ have the same $x$ and $y$ coordinates, and $q$ and $b'$ have the same $x$ and $y$ coordinates.
As with the 
algorithm presented in Section~\ref{BaseAlgTRlargeDiam},
this algorithm runs in $O(n^2m\log n)$ time.

\begin{figure}[hbt]
\centering
\includegraphics[width=.8\linewidth]{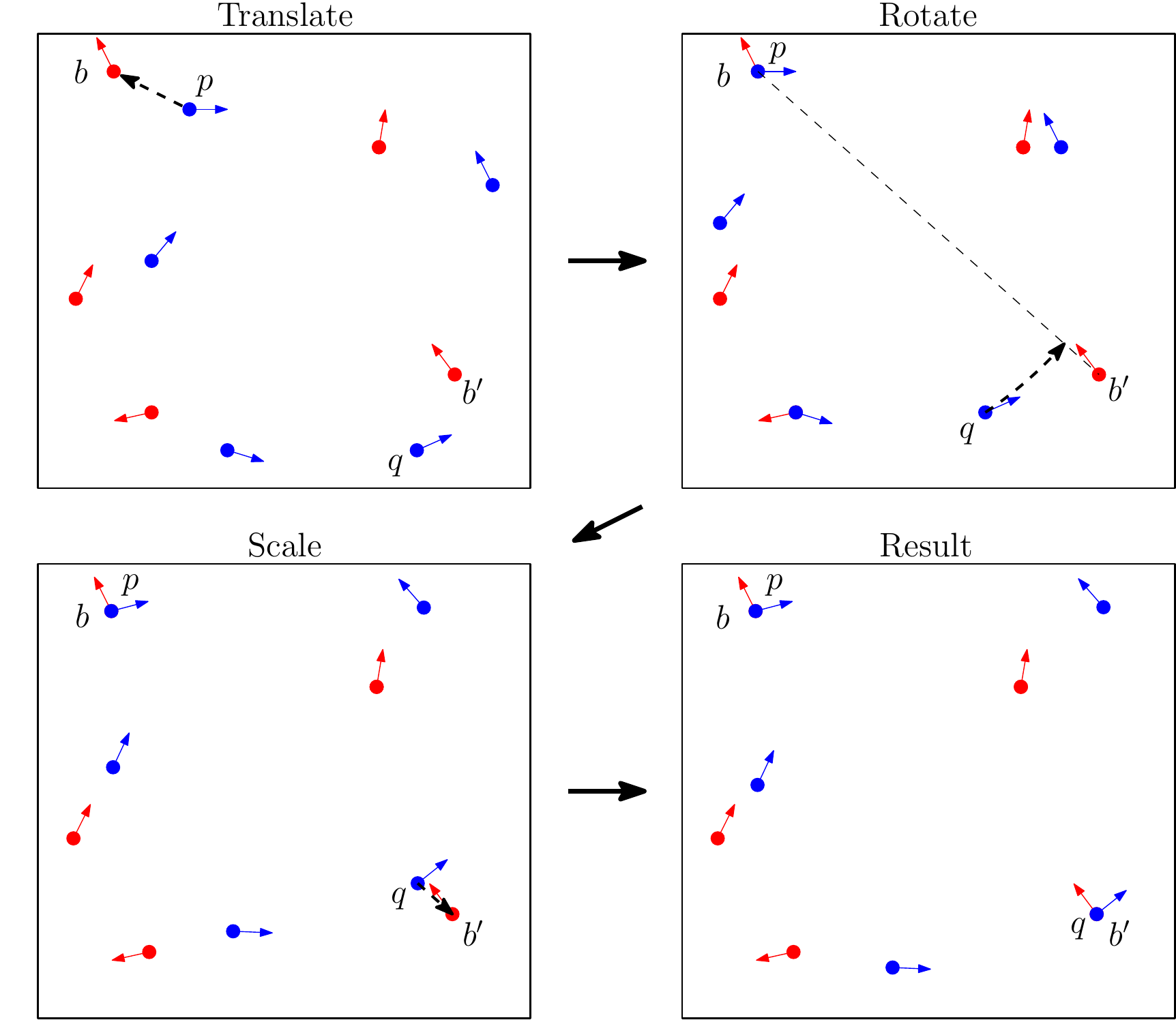}
\caption{Illustration of the translation, rotation and scaling steps of the base
approximation algorithm for translation, rotation and scaling in $O$ when diameter is
large.}
\label{fig:BaseAlgTRSlargeDiam}
\end{figure}

\begin{theorem}
\label{thm:BaseAlgTRSlargeDiam}
Let $h_{\textrm{opt}}$ be $h_i(E(P), B)$ where $E$ is the composition of functions in $
\mathcal{T} \cup \mathcal{R} \cup \mathcal{S}$ that attains the minimum of $h_i$. 
The algorithm above runs in time $O(n^2m\log n)$ and produces an approximation to 
$h_{\textrm{opt}}$ that is at most $(6 + \sqrt{2}(2 + \pi/D) + \epsilon)h_{\textrm{opt}}$ for $h_1$ and at most $(4 + \sqrt{2}(2 + \pi/D) + \epsilon)h_{\textrm{opt}}$ for $h_2$.
\end{theorem}

\begin{proof}
The additional $\epsilon$ terms come entirely from using approximate nearest neighbor queries, so it is sufficient to prove approximations which do not include the $\epsilon$ term using exact nearest neighbor queries. Particularly, we will prove a bound of $(6 + \sqrt{2}(2 + \pi/D))h_{\textrm{opt}}$ for $h_1$ and a bound of $(4 + \sqrt{2}(2 + \pi/D))h_{\textrm{opt}}$ for $h_2$.

Let $E$ be the composition of functions in $\mathcal{T} \cup \mathcal{R} \cup \mathcal{S}$ 
that attains the minimum of $h_i(E(P), B)$. Let $P'$ be $E(P)$. Because this algorithm is only an extension of the algorithm in Section~\ref{BaseAlgTRlargeDiam} we can follow the same logic as the proof of Theorem~\ref{thm:BaseAlgTRlargeDiam} to see that after the translation and rotation steps, each point  $p \in P'$ is at most $Ah_{\textrm{opt}}$ away from a background point $b \in B$ where $A = 6 + \sqrt{2}\pi/D$ for $h_1$ and $A = 2 + \sqrt{2}(2 + \pi/D)$ for $h_2$. Now we need only look at how much scaling increases the distance our points have moved. 

If $p,q \in P'$ are our diametric points after translation and rotation, and $p',q' \in B$ are the closest background points to the optimal position of $p$ and $q$ respectively, then let us define the point $q_t$ as the position of $q$ after translation, but prior to the rotation step. Now it is important to see that the points $q$, $q'$ and $q_t$ are three vertices of an isosceles trapezoid  where the line segment $q_tq'$ is a diagonal of the trapezoid and the line segment $qq_t$ is a base of the trapezoid. This situation is depicted in Figure~\ref{fig:baseTRStrap}. The length of the line segment $qq'$ is equal to the distance that $q$ will move when we scale $P'$ so that $q$ and $q'$ share the same position. Because $qq'$ is a leg of the trapezoid, the length of that leg can be no more than the length of the diagonal $q_tq'$. In the proof of Theorem~\ref{thm:BaseAlgTRlargeDiam}, we showed that $q_t$ is at most $2h_{\textrm{opt}}$ away from $q'$ so this implies that the distance $q$ moves from scaling is at most $2h_{\textrm{opt}}$.

Point $q$ is the farthest point away from the point $p$ that is the
center for scaling. Thus, no point moved farther as a result of
the scaling than $q$ did, with respect to $\mu_2$. For $\mu_1$ it is
possible that, if $q$ moved a distance $d$, another point could have
moved up to a distance $\sqrt{2}d$. Thus, we find that after scaling, any point in $P'$ is at most $(A+2\sqrt{2})h_{\textrm{opt}}$ and $(A+2)h_{\textrm{opt}}$ from its nearest background point for $\mu_1$ and $\mu_2$ respectively. Because this is a transformation that the algorithm checks, we are guaranteed at least this solution. Note that we assume $p'$ and $q'$ are not the same point. However if this is the case, then we know that $D \leq 2h_{\textrm{opt}}$ thus when we translate $p$ to $p'$ and scale $q$ down to $p'$ every point is within $(2\pi/D)h_{\textrm{opt}}$ of $p'$, which is a better approximation than the case where $p' \neq q'$ under our assumption that $D$ is large.
\end{proof}

\begin{figure}[hbt]
\centering
\includegraphics[width=0.7\linewidth]{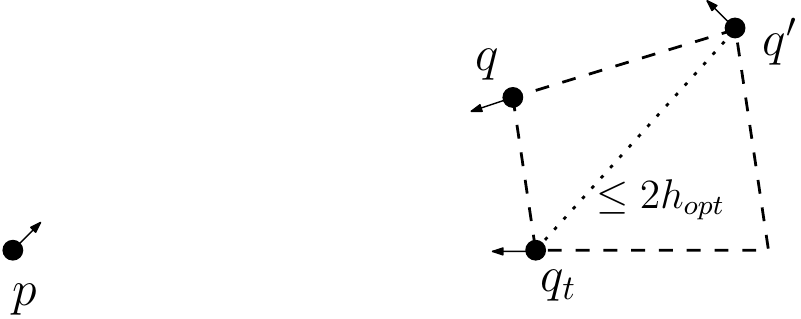}
\caption{Illustration of the points $q$, $q'$, and $q_t$ forming three of the corners of an isosceles trapezoid, as described in the proof of Theorem~\ref{thm:BaseAlgTRSlargeDiam}}
\label{fig:baseTRStrap}
\end{figure}

\subsection{A $(1+\epsilon)$-Approximation Algorithm Under Translation, Rotation and Scaling with Large Diameter}
\label{EpsAlgTRSlargeDiam}
In this subsection, we utilize the algorithm from 
Section~\ref{BaseAlgTRSlargeDiam} to achieve a 
$(1+\epsilon)$-approximation ratio when we allow 
translations, rotations, and scaling. Again, given two subsets of $O$, $P$ and $B$, with $|P| = m$ and $|B| = n$, our goal is to 
minimize $h_i(E(P), B)$ over all compositions $E$ of one or more functions in 
$\mathcal{T} \cup \mathcal{R} \cup \mathcal{S}$.
We perform the following steps.
\begin{enumerate}
\item Run BaseTranslateRotateScaleLarge($P,B$),
from Section~\ref{BaseAlgTRSlargeDiam},
to obtain an approximation $h_{apr} \leq A\cdot h_{\textrm{opt}}$.
\item For every $b \in B$, generate the point set $G_b = G(b, \frac{h_{apr}\epsilon}{A^2-A}, \lceil\frac{A^2-A}{\epsilon}\rceil)$ for $h_1$ or $G_b = G(b, \frac{\sqrt{2}h_{apr}\epsilon}{A^2-A}, \lceil\frac{A^2-A}{\sqrt{2}\epsilon}\rceil)$ for $h_2$.
Let $B'$ denote the resulting set.
\item Run BaseTranslateRotateScaleLarge($P,B'$),
from Section~\ref{BaseAlgTRSlargeDiam},
but use the set $B$ for the nearest-neighbor queries.

\end{enumerate}

This algorithm uses the base algorithm to give us 
an indication of what the optimal solution might be. 
We then use this approximation to generate a larger set of points from which to derive transformations to test. 
We next use this point set in the base algorithm when deciding which transformations to iterate over, while still using $B$ to compute nearest neighbors.
The running time is 
$O(A^8n^2m\log n)$, which is
$O(n^2m\log n)$ for constant $A$.

\begin{theorem}
\label{thm:EpsAlgTRSlargeDiam}
Let $h_{\textrm{opt}}$ be $h_i(E(P), B)$ where $E$ is the composition of functions in 
$\mathcal{T} \cup \mathcal{R} \cup \mathcal{S}$ that attains the minimum of $h_i$. The algorithm above runs in time 
$O(A^8n^2m\log n)$ and produces an approximation to $h_{\textrm{opt}}$ that is at most $(1+\epsilon)h_{\textrm{opt}}$ for
both $h_1$ and $h_2$.
\end{theorem}

\begin{proof}
Let $E$ be the composition of functions in $\mathcal{T} \cup \mathcal{R} \cup \mathcal{S}$ that attains the minimum of $h_i(E(P), B)$. 
Let $P'$ be $E(P)$. Then every point $q \in P'$ is at most $h_{\textrm{opt}}$ from the closest 
background point in $B$.
By running the base algorithm, we find $h_{apr}\leq Ah_{\textrm{opt}}$ where $A$ is the approximation ratio of the base algorithm. 
Now consider the point $b'\in B$ which is the closest background to some pattern point $p \in P$. The square which encompasses $G_{b'}$ has a side length of $2h_{apr}$. This guarantees that $p$, which is at most $h_{\textrm{opt}}$ away from $b'$, lies within this square. As we saw from Lemma~\ref{lem:cube}, this means that $p$ is at most $\frac{\epsilon h_{apr}}{A^2-A}$ away from its nearest neighbor in $G_{b'}$. Thus, if a transformation defined by the nearest points in $B$ would move our pattern points at most $(A-1)h_{\textrm{opt}}$ from their optimal position, then using the nearest points in $G_{b'}$ to define our transformation will move our points at most
\[
(A-1)\frac{\epsilon h_{apr}}{A^2-A} = \frac{\epsilon h_{apr}}{A} \leq \epsilon h_{\textrm{opt}}.
\]
Thus, the modified algorithm gives a solution that is at most $(1+\epsilon)h_{\textrm{opt}}$. 
\end{proof}

\subsection{Base Algorithm Under Translation, Rotation and Scaling with Small Diameter}
\label{BaseAlgTRSsmallDiam}
In this subsection, we present an alternative algorithm for solving the approximate oriented point-set pattern matching problem where we allow translations, rotations and scaling. This algorithm is an extension of the algorithm from Section~\ref{BaseAlgTRsmallDiam} and similarly provides a good approximation ratio when the diameter of our pattern set is small. 
Once again, given two subsets of $O$, $P$ and $B$, with $|P| = m$ and $|B| = n$, we wish to minimize $h_i(E(P), B)$ over all compositions $E$ of one or more functions in 
$\mathcal{T} \cup \mathcal{R}$. 
We perform the following algorithm:

\begin{center}
\rule{\columnwidth}{2pt}
\textbf{Algorithm} BaseTranslateRotateSmall($P,B$):
 \vspace*{-10pt}
\begin{algorithmic}
\STATE 
Find $p$ and $q$ in $P$ having the maximum value of $\|(x_p, y_p) - (x_q, y_q)\|_2$.
\FOR {every point $b \in B$}
\STATE \emph{$1^{\rm st}$ Pin:} 
Apply the translation, $T_v\in {\cal T}$, that takes $p$ to $b$,
and then
apply the rotation,
$R_{p,\theta}$, that makes $p$, $b$ have the same orientation.
\STATE 
Let $P'$ denote the transformed pattern set, $P$.
\FOR {each point $p$ in $P'$ and each $b' \in B$}
\STATE \emph{$2^{\rm nd}$ pin:} 
Apply the scaling, $S_{p,s}$, 
so that $\|(x_p, y_p) - (x_q, y_q)\|_2 = \|(x_b, y_b) - (x_{b'}, y_{b'})\|_2$%
\STATE
Let $P''$ denote the transformed pattern set.
\FOR {every $q \in P''$}
\STATE \emph{Query step:}
Find a nearest-neighbor of $q$ in $B$ using the $\mu_i$ metric, and update
a candidate Hausdorff distance accordingly.
\ENDFOR
\ENDFOR
\STATE \textbf{return} the smallest candidate Hausdorff distance
found as the smallest Hausdorff distance, $h_i(S_{p,s}(R_{p,\theta}(T_v(P))), B)$.
\ENDFOR 
\end{algorithmic}
\vspace*{-4pt}
\rule{\columnwidth}{2pt}
\end{center}

This algorithm extends the algorithm from Section~\ref{BaseAlgTRsmallDiam} by scaling the point set for so that $p$, $q$, and $b'$ form the vertices of an isosceles triangle. This requires a factor of $n$ more transformations to be computed.
Thus, the running time of this algorithm is $O(n^2m\log n)$.

\begin{figure}[hbt]
\centering
\includegraphics[width=.94\linewidth]{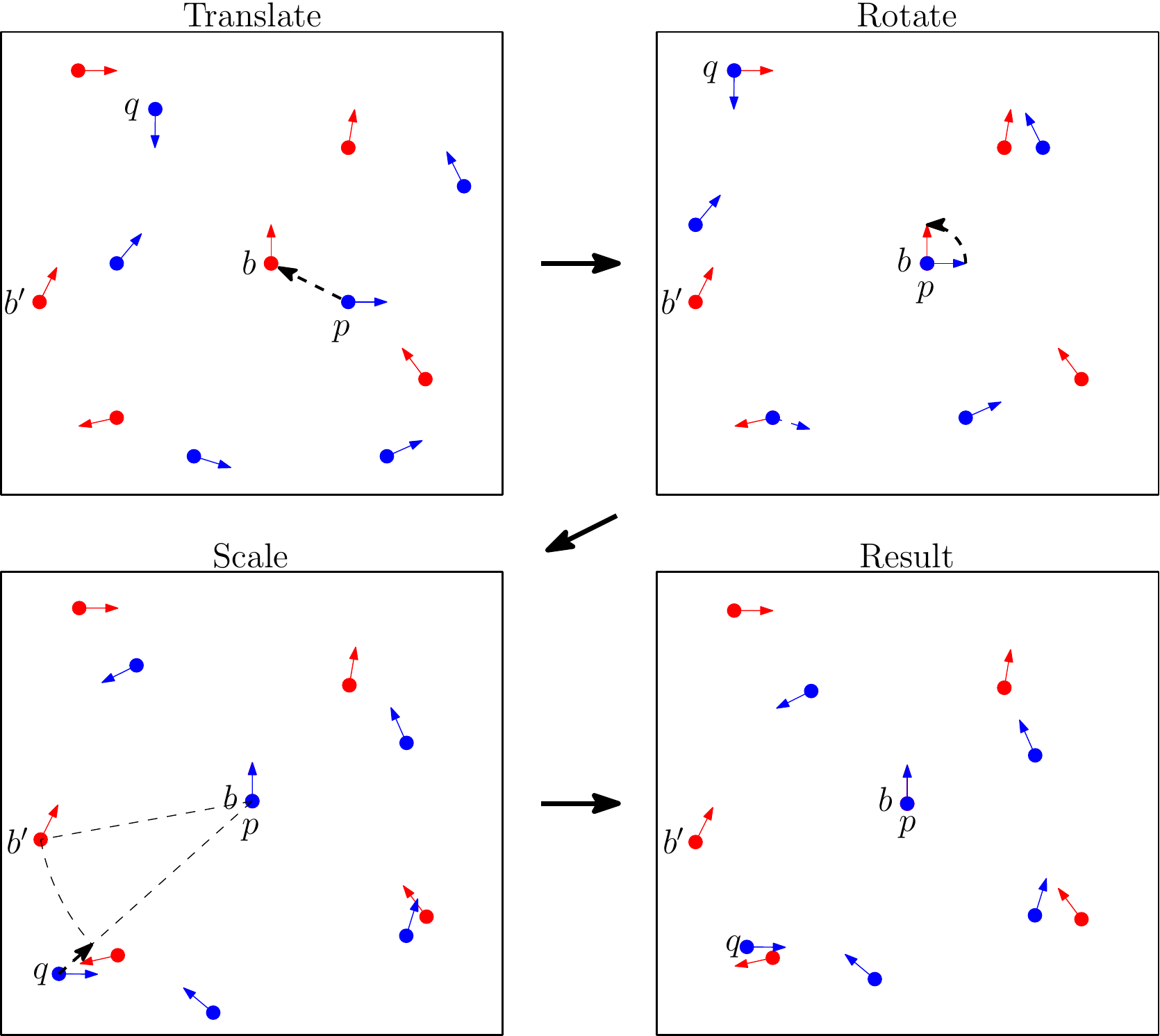}
\caption{Illustration of the translation, rotation and scaling steps of the base
approximation algorithm for translation, rotation and scaling in $O$ when diameter is
small.}
\label{fig:BaseAlgTRSsmallDiam}
\end{figure}

\begin{theorem}
\label{thm:BaseAlgTRSsmallDiam}
Let $h_{\textrm{opt}}$ be $h_i(E(P), B)$ where $E$ is the composition of functions in $
\mathcal{T} \cup \mathcal{R} \cup \mathcal{S}$ that attains the minimum of $h_i$. The 
algorithm above runs in time $O(n^2m\log n)$ and produces an approximation to $h_{\textrm{opt}}$ that is at most $((2 + 2\sqrt{2})(1 + D) + \epsilon)h_{\textrm{opt}}$ for $h_1$ and at most $(4 + 2D + \epsilon)h_{\textrm{opt}}$ for $h_2$.
\end{theorem}

\begin{proof}
The additional $\epsilon$ terms come entirely from using approximate nearest neighbor queries, so it is sufficient to prove approximations which do not include the $\epsilon$ term using exact nearest neighbor queries. Particularly, we will prove a bound of $(6 + \sqrt{2}(2 + \pi/D))h_{\textrm{opt}}$ for $h_1$ and a bound of $(4 + \sqrt{2}(2 + \pi/D))h_{\textrm{opt}}$ for $h_2$.

Let $E$ be the composition of functions in $\mathcal{T} \cup \mathcal{R} \cup \mathcal{S}$ 
that attains the minimum of $h_i(E(P), B)$. Let $P'$ be $E(P)$. Because this algorithm is only an extension of the algorithm in Section~\ref{BaseAlgTRsmallDiam} we can follow the same logic as the proof of Theorem~\ref{thm:BaseAlgTRsmallDiam} to see that after the translation and rotation steps, each point  $p \in P'$ is at most $Ah_{\textrm{opt}}$ away from a background point $b \in B$ where $A = 2 + \sqrt{2}D$ for $h_1$ and $A = 2 + D$ for $h_2$. Now we need only look at how much scaling increases the distance our points have moved. 

If $p,q \in P'$ are our diametric points after translation and rotation, and $p',q' \in B$ are the closest background points to the optimal position of $p$ and $q$ respectively, then let us define the point $q_s$ as the position of $q$ after scaling. The points $q$, $q'$ and $q_s$ are three vertices of an isosceles trapezoid  where the line segment $qq'$ is a diagonal of the trapezoid and the line segment $q_sq'$ is a base of the trapezoid. The length of the line segment $qq_s$ is equal to the distance that $q$ will move when we scale $P'$. Because $qq_s$ is a leg of the trapezoid, the length of that leg can be no more than the length of the diagonal $qq'$. In the proof of Theorem~\ref{thm:BaseAlgTRsmallDiam}, we showed that $q$ is at most $Ah_{\textrm{opt}}$ away from $q'$ so this implies that the distance $q$ moves from scaling is at most $Ah_{\textrm{opt}}$.

Point $q$ is the farthest point away from the point $p$ which is the center of our scaling. Thus, no point moves farther as a result of the scaling than $q$ does, with respect to $\mu_2$. For $\mu_1$ it is possible that, if $q$ moved a distance $d$, another point could have moved up to a distance $\sqrt{2}d$. Thus we find that after scaling, any point in $P'$ is at most $(1+\sqrt{2})Ah_{\textrm{opt}}$ and $2Ah_{\textrm{opt}}$ from its nearest background point for $\mu_1$ and $\mu_2$ respectively. Because this is a transformation that the algorithm checks, we are guaranteed at least this solution.
\end{proof}

\subsection{A $(1+\epsilon)$-Approximation Algorithm Under Translation, Rotation and Scaling with Small Diameter}
\label{EpsAlgTRSsmallDiam}
In this subsection, we utilize the algorithm from 
Section~\ref{BaseAlgTRSsmallDiam} to achieve a $(1+\epsilon)$-approximation ratio 
when we allow translations, rotations, and scalings. 
Again, given two subsets of $O$, $P$ and $B$, with $|P| = m$ and $|B| = n$, our goal is to 
minimize $h_i(E(P), B)$ over all compositions $E$ of one or more functions in 
$\mathcal{T} \cup \mathcal{R} \cup \mathcal{S}$.
We perform the following steps.
\begin{enumerate}
\item Run BaseTranslateRotateScaleSmall($P,B$),
from Section~\ref{BaseAlgTRSsmallDiam} to obtain an approximation 
$h_{apr} \leq A\cdot h_{\textrm{opt}}$.
\item For every $b \in B$, generate the point set $G_b = G(b, \frac{h_{apr}\epsilon}{2(A^2-A)}, \lceil\frac{2(A^2-A)}{\epsilon}\rceil)$ for $h_1$ or $G_b = G(b, \frac{h_{apr}\epsilon}{A^2-A}, \lceil\frac{A^2-A}{\epsilon}\rceil)$ for $h_2$.
Let $B' = \bigcup_{b \in B}G_b$ denote the resulting set of points.
\item For every $b' \in B'$, generate the point set $C_{b'} = C(b', \frac{2(A^2-A)}{\pi h_{apr}\epsilon})$ for $h_1$ or $C_{b'} = C(b', \frac{\sqrt{2}(A^2-A)}{\pi h_{apr}\epsilon})$ for $h_2$.
Let $B''$ denote the resulting set of points.
\item Run BaseTranslateRotateScaleSmall($P,B''$),
but use the points in $B$ for nearest-neighbor queries.

\end{enumerate}

This algorithm uses the 
base algorithm to give us an indication of what the optimal solution might be. 
We use this approximation to generate a larger set of points from which to derive transformations to test, but this time we also generate a number of different orientations for those points as well. 
We then use this point set in the base algorithm when deciding which transformations to iterate over, while still using $B$ to compute nearest neighbors.
The running time of this algorithm is $O(A^{12}n^2m\log n)$.

\begin{theorem}
\label{thm:EpsAlgTRSsmallDiam}
Let $h_{\textrm{opt}}$ be $h_i(E(P), B)$ where $E$ is the composition of functions in 
$\mathcal{T} \cup \mathcal{R}$ that attains the minimum of $h_i$. The algorithm above runs in time 
$O(A^{12}n^2m\log n)$ and produces an approximation to $h_{\textrm{opt}}$ that is at most $(1+\epsilon)h_{\textrm{opt}}$ for
both $h_1$ and $h_2$.
\end{theorem}

\begin{proof}
Let $E$ be the composition of functions in $\mathcal{T} \cup \mathcal{R} \cup \mathcal{S}$ that attains the minimum of $h_i(E(P), B)$. 
Let $P'$ be $E(P)$. Then every point $q \in P'$ is at most $h_{\textrm{opt}}$ from the closest 
background point in $B$.
By running the base algorithm, we find $h_{apr}\leq Ah_{\textrm{opt}}$ where $A$ is the approximation ratio of the base algorithm. Now consider the point $b'\in B$ which is the closest background to some pattern point $p \in P$. The square which encompasses $G_{b'}$ has a side length of $2h_{apr}$. This guarantees that $p$, which is at most $h_{\textrm{opt}}$ away from $b'$, lies within this square. As we saw from Lemma~\ref{lem:cube}, this means that $p$ is at most $\frac{\epsilon h_{apr}}{2(A^2-A)}$ away from its nearest neighbor $g$ in $G_{b'}$ with respect to the $L1$-norm, and at most $\frac{\epsilon h_{apr}}{\sqrt{2}(A^2-A)}$ with respect to the $L2$-norm. For this point $g$, there are a number of points in $C_g$ which are at the same position but with different orientation. For some point $c$ in $C_g$, the orientation of point $p$ is within an angle of at most $\frac{h_{apr}\epsilon}{2(A^2-A)}$ for $h_1$ and at most $\frac{h_{apr}\epsilon}{\sqrt{2}(A^2-A)}$ for $h_2$. If we combine together the maximum difference in position between $p$ and $c$, and the maximum difference in orientation between $p$ and $c$, then we see that for both $\mu_1$ and $\mu_2$, the distance between $p$ and $c$ is at most $\frac{h_{apr}\epsilon}{A^2-A}$. As we explain at the beginning of this section, if a transformation defined by the nearest points in $B$ would move our pattern points at most $(A-1)h_{\textrm{opt}}$ from their optimal position, then using the nearest points in $C_{g}$ to define our transformation will move our points at most $(A-1)\frac{\epsilon h_{apr}}{A^2-A} = \frac{\epsilon h_{apr}}{A} \leq \epsilon h_{\textrm{opt}}$. Thus the modified algorithm gives a solution that is at most $(1+\epsilon)h_{\textrm{opt}}$. 
\end{proof}

As with our methods for translation and rotation, we can compute in advance
whether we should run our algorithm for large diameter point sets or our
algorithm for small diameter point sets. For $h_1$, we compare the expressions $6 + \sqrt{2}(2 + \pi/D)$ and $(2+2
\sqrt{2})(1+D)$, and we find that the two expressions are equal at $D^* \approx 1.46$. For $h_2$, we compare $4 + \sqrt{2}(2 + \pi/D)$ and $4+2D$ to find that they are equal at $D^*  \approx 2.36$. Using $D^*$ as the deciding value allows us to then find
a transformation in $\mathcal{T}\cup\mathcal{R}\cup \mathcal{S}$ that achieves
a $(1+\epsilon)$-approximation, for any constant $\epsilon>0$,
in $O(n^2m\log n)$ time.

\section{Experiments}

In reporting the results of our experiements, we use the following labels for the algorithms: 
\begin{itemize}
\item
\emph{GR}: the non-oriented translation and rotation algorithm from Goodrich {\it et al.}~\cite{goodrich1999approximate},
\item
\emph{LD}$_{h_1/h_2}$: the base version of the large diameter algorithm using either the $h_1$ or $h_2$ distance metric, 
\item
\emph{SD}$_{h_1/h_2}$: the base version of the small diameter algorithm using either the $h_1$ or $h_2$ distance metric. 
\end{itemize}

These algorithms were implemented in C++ (g++ version 4.8.5) and run on a Quad-core Intel Xeon 3.0GHz CPU E5450 with 32GB of RAM on 64-bit CentOS Linux 6.6. 

\subsection{Accuracy Comparison} We tested the ability of each algorithm to identify the orginal point set after it had been slightly perturbed. From set of randomly generated oriented background point sets, one was selected and a random subset of the points in the set were shifted and rotated by a small amount. Each algorithm was used to match this modified pattern against each of the background point sets and it was considered a success if the background set from which the pattern was derived had the smallest distance (as determined by each algorithm's distance metric). Figure~\ref{fig:accuracy} shows the results of this experiment under two variables: the number of background sets from which the algorithms could choose, and the size of the background sets. Each data point is the percentage of successes across 1000 different pattern sets.

\begin{figure}[hbt]
\centering
\includegraphics[width=\linewidth]{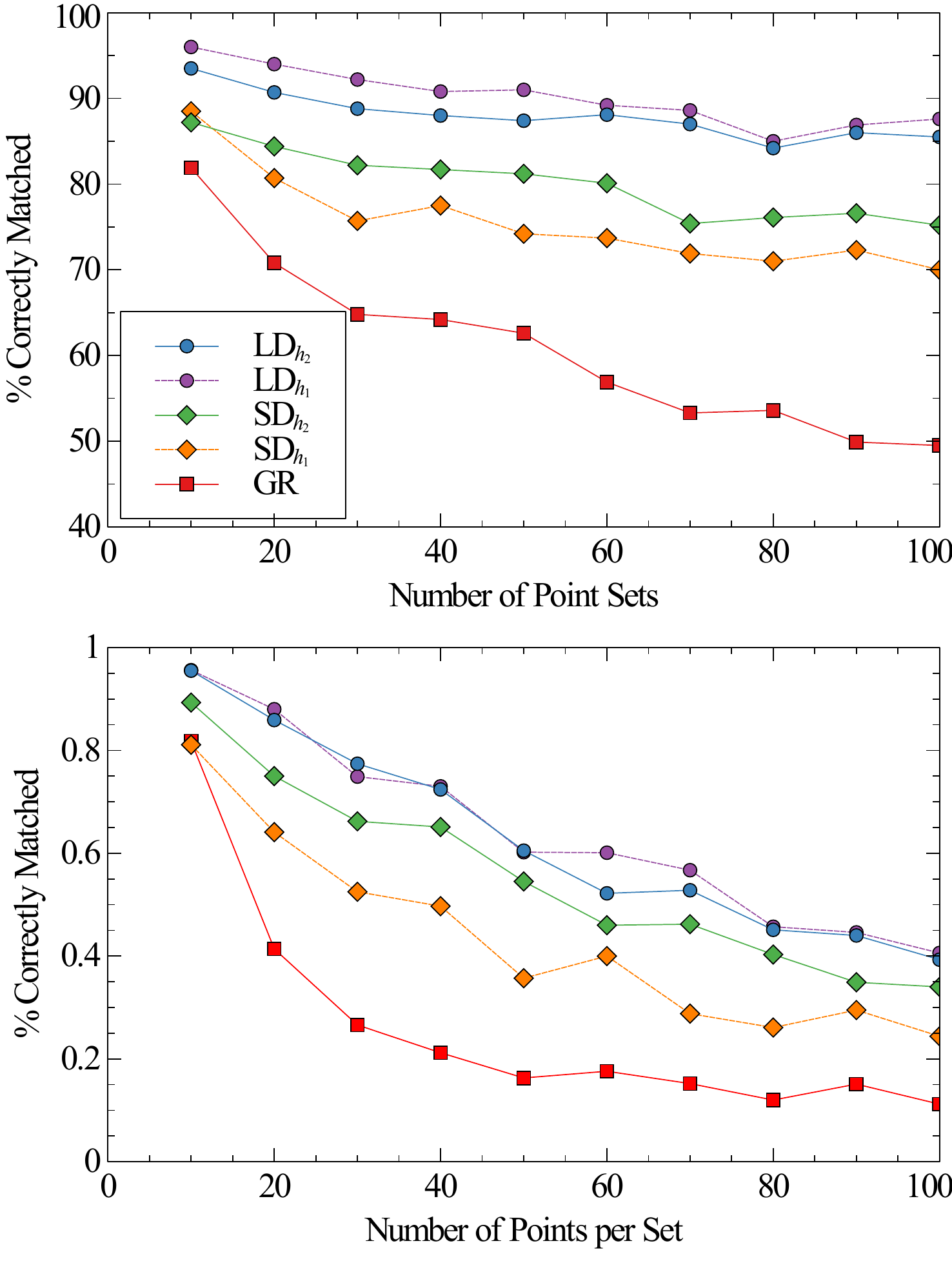}
\caption{Results of Accuracy Comparison}
\label{fig:accuracy}
\end{figure}

In every case, the oriented algorithms are more successful at identifying the origin of the pattern than GR. LD  was also more successful for each distance metric than SD.

\subsection{Performance Comparison} We also compared the performance of the LD and SD algorithms against GR as we increased the pattern size and the background size. The most significant impact of increasing the background size is that the number of nearest neighbor queries increase, and thus the performance in this case is dictated by quality of the nearest neighbor data structure used. Therefore in Figure~\ref{fig:nnq} we use the number of nearest neighbor queries as the basis for comparing performance. As the FD and GR algorithms only differ in how the nearest neighbor is calculated, they both perform the same number of queries while the SD algorithm performs significantly fewer nearest neighbor queries.

For pattern size, we compared running time and the results are shown in Figure~\ref{fig:time}. In this case, LD is slower than GR, while SD is signifcantly faster than either of the others. 

   \begin{figure}[hbt]
\centering
\includegraphics[width=\linewidth]{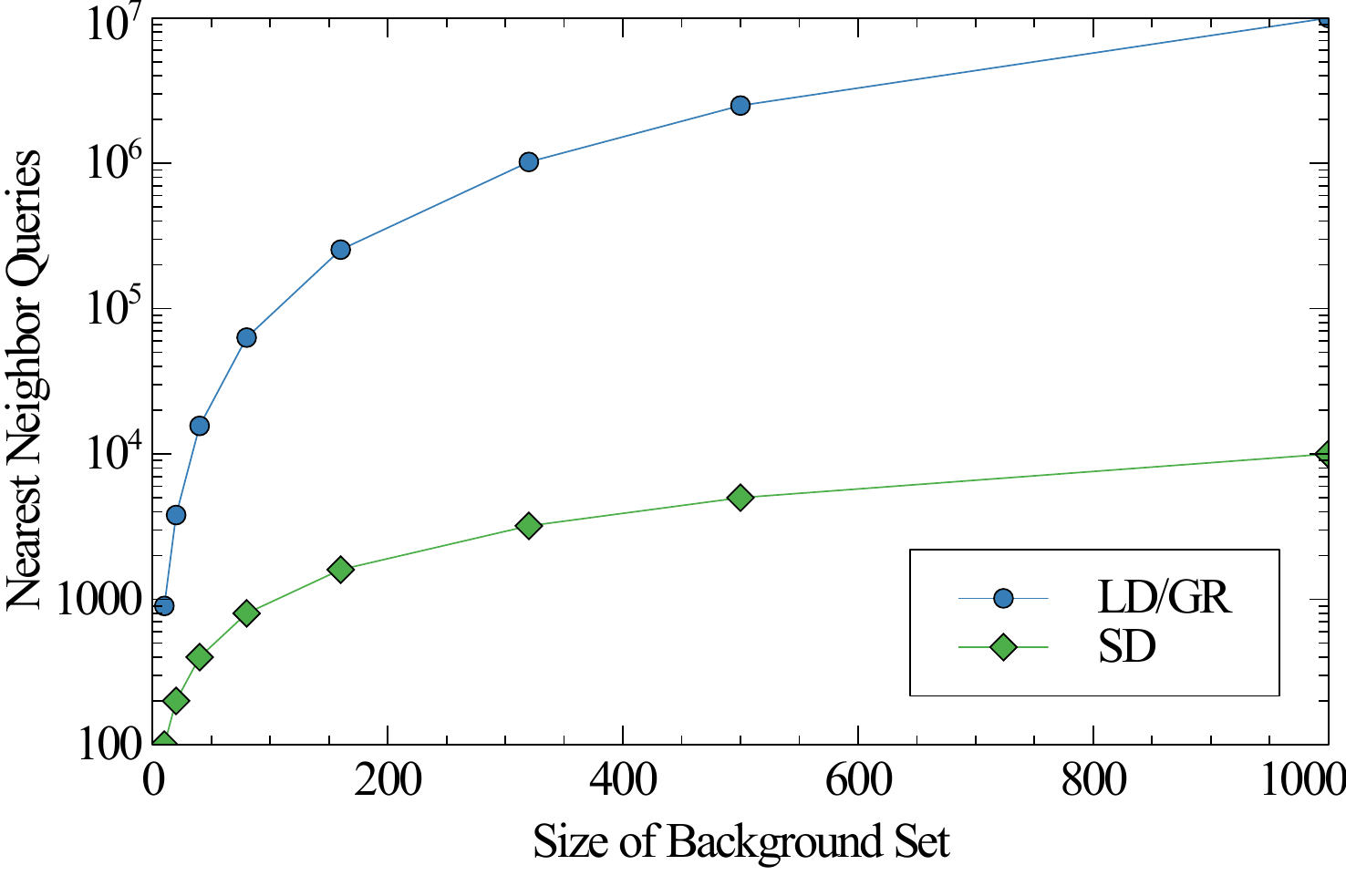}
\caption{Comparison of nearest neighbor queries as function of background size}
\label{fig:nnq}
\end{figure}

\begin{figure}[hbt]
\centering
\includegraphics[width=\linewidth]{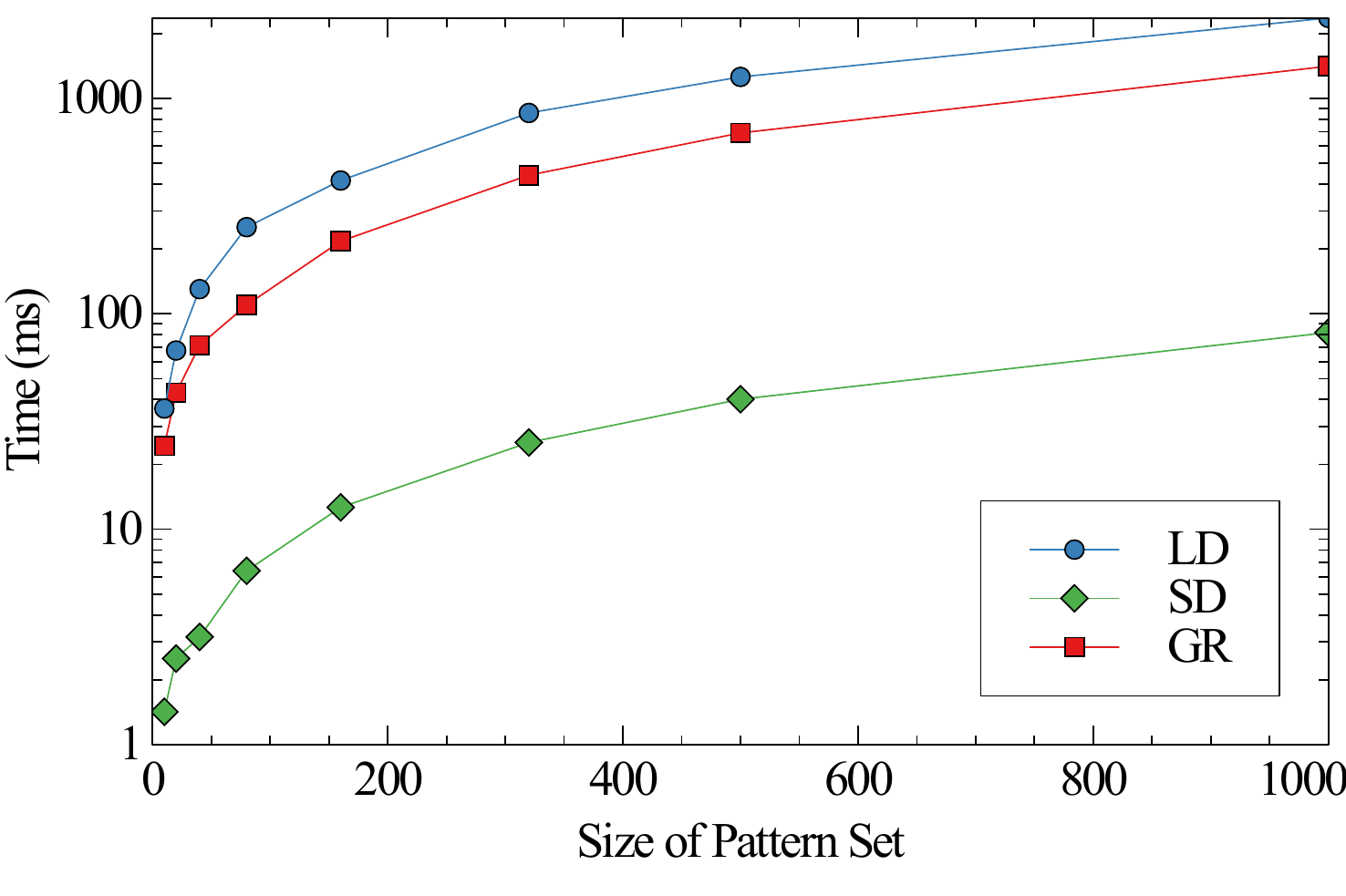}
\caption{Comparison of running time as a function of pattern size}
\label{fig:time}
\end{figure}

\section{Conclusion} 
We present distance metrics that can be used to measure the 
similarity between two point sets with orientations and we also provided fast 
algorithms that guarantee close approximations of an optimal 
transformation.
In the appendices,
we provide additional algorithms for 
other types of transformations and we also provide
results of experiments.


\section*{Acknowledgments}

This work was supported in by the NSF under grants 1526631, 1618301, and 1616248, and by DARPA under agreement no. AFRL FA8750-15-2-0092. The views expressed are those of the authors and do not reflect the official policy or position of the Department of Defense or the U.S. Government.

\clearpage
\bibliographystyle{abbrv}

\bibliography{point-set-pattern-matching-archive}

\begin{thebibliography}{10}

\bibitem{alt1999discrete}
H.~Alt and L.~J. Guibas.
\newblock Discrete geometric shapes: Matching, interpolation, and
  approximation.
\newblock {\em Handbook of computational geometry}, 1:121--153, 1999.

\bibitem{arya1998optimal}
S.~Arya, D.~M. Mount, N.~S. Netanyahu, R.~Silverman, and A.~Y. Wu.
\newblock An optimal algorithm for approximate nearest neighbor searching fixed
  dimensions.
\newblock {\em Journal of the ACM (JACM)}, 45(6):891--923, 1998.

\bibitem{cardoze1998pattern}
D.~E. Cardoze and L.~J. Schulman.
\newblock Pattern matching for spatial point sets.
\newblock In {\em Foundations of Computer Science, 1998. Proceedings. 39th
  Annual Symposium on}, pages 156--165. IEEE, 1998.

\bibitem{chew1997geometric}
L.~P. Chew, M.~T. Goodrich, D.~P. Huttenlocher, K.~Kedem, J.~M. Kleinberg, and
  D.~Kravets.
\newblock Geometric pattern matching under euclidean motion.
\newblock {\em Computational Geometry}, 7(1):113--124, 1997.

\bibitem{cho2008improved}
M.~Cho and D.~M. Mount.
\newblock Improved approximation bounds for planar point pattern matching.
\newblock {\em Algorithmica}, 50(2):175--207, 2008.

\bibitem{gavrilov1999geometric}
M.~Gavrilov, P.~Indyk, R.~Motwani, and S.~Venkatasubramanian.
\newblock Geometric pattern matching: A performance study.
\newblock In {\em Proceedings of the fifteenth annual symposium on
  Computational geometry}, pages 79--85. ACM, 1999.

\bibitem{goodrich1999approximate}
M.~T. Goodrich, J.~S. Mitchell, and M.~W. Orletsky.
\newblock Approximate geometric pattern matching under rigid motions.
\newblock {\em IEEE Transactions on Pattern Analysis and Machine Intelligence},
  21(4):371--379, 1999.

\bibitem{indyk1999geometric}
P.~Indyk, R.~Motwani, and S.~Venkatasubramanian.
\newblock Geometric matching under noise: Combinatorial bounds and algorithms.
\newblock In {\em SODA}, pages 457--465, 1999.

\bibitem{jain1997identity}
A.~K. Jain, L.~Hong, S.~Pankanti, and R.~Bolle.
\newblock An identity-authentication system using fingerprints.
\newblock {\em Proceedings of the IEEE}, 85(9):1365--1388, 1997.

\bibitem{JEA20051672}
T.-Y. Jea and V.~Govindaraju.
\newblock A minutia-based partial fingerprint recognition system.
\newblock {\em Pattern Recognition}, 38(10):1672--1684, 2005.

\bibitem{jiang906252}
X.~Jiang and W.-Y. Yau.
\newblock Fingerprint minutiae matching based on the local and global
  structures.
\newblock In {\em Proceedings 15th International Conference on Pattern
  Recognition. ICPR-2000}, volume~2, pages 1038--1041, 2000.

\bibitem{kulkarni2006orientation}
J.~V. Kulkarni, B.~D. Patil, and R.~S. Holambe.
\newblock Orientation feature for fingerprint matching.
\newblock {\em Pattern Recognition}, 39(8):1551--1554, 2006.

\bibitem{maltoni2009handbook}
D.~Maltoni, D.~Maio, A.~Jain, and S.~Prabhakar.
\newblock {\em Handbook of Fingerprint Recognition}.
\newblock Springer Science \& Business Media, 2009.

\bibitem{preparatacomputational}
F.~P. Preparata and M.~I. Shamos.
\newblock Computational geometry: an introduction.
\newblock {\em Springer-Verlag, New York, NY}, 1985.

\bibitem{qi2005fingerprint}
J.~Qi, S.~Yang, and Y.~Wang.
\newblock Fingerprint matching combining the global orientation field with
  minutia.
\newblock {\em Pattern Recognition Letters}, 26(15):2424--2430, 2005.

\bibitem{ratha2007automatic}
N.~Ratha and R.~Bolle.
\newblock {\em Automatic Fingerprint Recognition Systems}.
\newblock Springer Science \& Business Media, 2007.

\bibitem{tico2003fingerprint}
M.~Tico and P.~Kuosmanen.
\newblock Fingerprint matching using an orientation-based minutia descriptor.
\newblock {\em IEEE Transactions on Pattern Analysis and Machine Intelligence},
  25(8):1009--1014, 2003.

\bibitem{veltkamp2001shape}
R.~C. Veltkamp.
\newblock Shape matching: similarity measures and algorithms.
\newblock In {\em Shape Modeling and Applications, SMI 2001 International
  Conference on.}, pages 188--197. IEEE, 2001.

\bibitem{Xu09}
H.~Xu, R.~N.~J. Veldhuis, T.~A.~M. Kevenaar, and T.~A. H.~M. Akkermans.
\newblock A fast minutiae-based fingerprint recognition system.
\newblock {\em IEEE Systems Journal}, 3(4):418--427, Dec 2009.

\end{thebibliography}


\clearpage

\end{document}